\documentclass[aps,pra,twocolumn,longbibliography,superscriptaddress]{revtex4-1}

\usepackage{amsmath}
\usepackage{amsthm}
\usepackage{amssymb}
\usepackage{mathtools} 
\usepackage{bm}
\usepackage{braket}
\usepackage{lineno}
\usepackage{pifont}
\usepackage{bbm}
\usepackage{bbding}

\usepackage[svgnames]{xcolor}
\usepackage[many]{tcolorbox}
\newtcolorbox{mybox}[1][]{
  drop shadow southeast,
  enhanced,colback=red!5!white,colframe=black!75!, width=0.36\textwidth,
  #1
}
\newtcolorbox{fancybox}[1][]{
  enhanced,drop fuzzy midday shadow,
  boxrule=1pt,arc=4pt,boxsep=0pt,
  left=.5em,right=.5em,top=1ex,bottom=1ex,
  colback=olive, width=0.46\textwidth,#1
}

\newtcbox{\testbox}[1][red]
  {on line, enhanced, drop shadow southeast,arc = 0pt, outer arc = 0pt,
    colback = #1!5!white, colframe = red!75!black,
    boxsep = 0pt, left = 1pt, boxrule = 0pt}
    
\newtheorem{theorem}{Theorem}

\newtheorem{lemma}{Lemma}

\newtheorem{definition}{Definition}

\usepackage{graphicx}
\usepackage{subfig}
\usepackage{booktabs}

\usepackage[hidelinks]{hyperref}
\usepackage{url}
 
\usepackage{algorithm}
\usepackage{algpseudocode}
 
\usepackage{caption}
\usepackage{adjustbox}

\DeclareMathOperator{\Tr}{Tr}

\DeclareMathOperator{\RY}{R_Y}
\DeclareMathOperator{\RZ}{R_Z}
\DeclareMathOperator{\Ber}{Ber}

\DeclareMathOperator{\Z}{Z}
\newcommand{\barL}{\bar{\mathcal{L}}}
  
\DeclareMathOperator{\X}{X}
\DeclareMathOperator{\Y}{Y}
\DeclareMathOperator{\I}{I}

\begin{document}
\title{Accelerating variational quantum algorithms with multiple quantum processors}
 
\author{Yuxuan Du}
\thanks{Corresponding author, duyuxuan123@gmail.com}
\affiliation{JD Explore Academy}
\author{Yang Qian}
\thanks{This work was done when he was a research intern at JD Explore Academy.} 
\affiliation{School of Computer Science, The University of Sydney}
\affiliation{JD Explore Academy}
\author{Dacheng Tao}
 \affiliation{JD Explore Academy}

\begin{abstract}
Variational quantum algorithms (VQAs) have the potential of utilizing near-term quantum machines to gain certain computational advantages over classical methods. Nevertheless, modern VQAs suffer from cumbersome computational overhead, hampered by the tradition of employing a solitary quantum processor to handle large-volume data. As such, to better exert the superiority of VQAs, it is of great significance to improve their runtime efficiency. Here we devise an efficient distributed optimization scheme, called QUDIO, to address this issue. Specifically, in QUDIO, a classical central server  partitions the learning problem into multiple subproblems and allocate them to multiple local nodes where each of them consists of a quantum processor and a classical optimizer.  During the training procedure, all local nodes proceed parallel optimization and the classical server synchronizes  optimization information among local nodes timely. In doing so, we prove a sublinear convergence rate of QUDIO in terms of the number of global iteration under the ideal scenario, while the system imperfection may incur divergent optimization. Numerical results on standard benchmarks demonstrate that QUDIO can surprisingly achieve a superlinear runtime speedup with respect to the number of local nodes. Our proposal can be readily mixed with other advanced VQAs-based techniques  to narrow the gap between the  state of the art and applications with quantum advantage.
\end{abstract}
\maketitle    

\section{Introduction}
Deep learning techniques have penetrated the world during past decades \cite{goodfellow2016deep}. Prototypical applications comprise using deep neural networks (DNNs) to facilitate online shopping  \cite{zhang2019deep}, to accelerate molecule design  \cite{senior2020improved}, and to enhance language translation  \cite{devlin2019bert}. Notably, the success of deep learning heavily relies on distributed hardware and distributed optimization techniques \cite{Jeffrey2012} in the sense that multiple GPU cards are employed to collaboratively process the same learning task. For instance, to ensure that the training of deep bidirectional transformers (BERT) can be completed within a reasonable time (e.g., $4$ days), in total $64$ GPU cards are used to conduct the distributed optimization \cite{devlin2019bert}.  However, the price to pay is  ten thousands dollars and emitting $1438$ lbs of carbon dioxide \cite{strubell2019energy}.  This considerable resource-consumption signifies that deep learning models will become hard to develop and optimize  when the problem size is continuously enlarged. Therefore, it is highly demanded to seek novel techniques to resolve complicated problems in a fast, economic, and environmentally friendly way. 
 
The oath of quantum computing is to accomplish certain tasks beyond the reach of classical computers \cite{biamonte2017quantum,feynman1982simulating,harrow2017quantum}. During past years, big breakthrough has been achieved towards this goal, e.g., a demonstration of the quantum supremacy experiment in the task of sampling the output of a pseudo-random quantum circuit \cite{arute2019quantum}. Among various quantum computational models, variational quantum algorithms (VQAs) \cite{benedetti2019parameterized,bharti2021noisy,cerezo2020variational2,du2018expressive,endo2021hybrid}, which are formed by parameterized quantum circuits (PQCs) and a classical optimizer as shown in the left panel of Fig.~\ref{fig:sheme}, have attracted great attention from industry and academia. The popularity of VQAs origins from the versatility of PQCs, which guarantees their efficient implementations on the noisy intermediate-scale quantum (NISQ) machines \cite{preskill2018quantum}, as well as theoretical evidence of quantum superiority \cite{abbas2020power,banchi2021generalization,bu2021effects,caro2021encoding,du2021efficient,huang2021information,huang2021power,wang2021towards,wu2021expressivity}. Moreover, prior studies have exhibited the progress of VQAs of accomplishing diverse learning tasks, e.g., machine learning issues such as data classification \cite{Du_2021_grover,havlivcek2019supervised,mitarai2018quantum,schuld2019quantum,wang2020quantum} and image generation \cite{huang2020experimental,rudolph2020generation,zhu2019training}, combinatorial optimization \cite{crooks2018performance,farhi2014quantum,hadfield2019quantum,harrigan2021quantum}, finance \cite{alcazar2020classical,coyle2021quantum,hodson2019portfolio}, quantum information processing \cite{beckey2020variational,cerezo2020variational,larose2019variational}, quantum chemistry and material sciences \cite{bauer2016hybrid,google2020hartree,o2019calculating,peruzzo2014variational}, and particle physics \cite{avkhadiev2020accelerating,kokail2019self}.    
  
Despite the tantalizing achievements, modern VQAs generally are plagued by the expensive or even unaffordable execution time for large-volume data, hampered by the regulation such that only a \textit{single} quantum chip is employed in  optimization. For concreteness, let us recall the machinery of quantum neural networks (QNNs), as a crucial subclass of VQAs, when dealing with classification tasks \cite{beer2020training,du2020learnability,farhi2018classification,havlivcek2019supervised,mitarai2018quantum,schuld2019quantum}. To correctly classify $n$ training examples, the classical optimizer iteratively feeds each example to PQCs and then leverages the discrepancy between the obtained $n$ predictions of QNN and $n$ labels to conduct the gradient-based optimization. In this way, the classical optimizer needs to query the single quantum processor at least $O(n)$ times to complete one iteration. Such overhead prohibits the applicability of QNNs for large $n$ \cite{qian2021dilemma}. With this regard, it is of great importance to enhance the computational efficiency of VQAs, as the necessary condition to pursue quantum advantages.

\begin{figure*}
\centering
\includegraphics[width=0.98\textwidth]{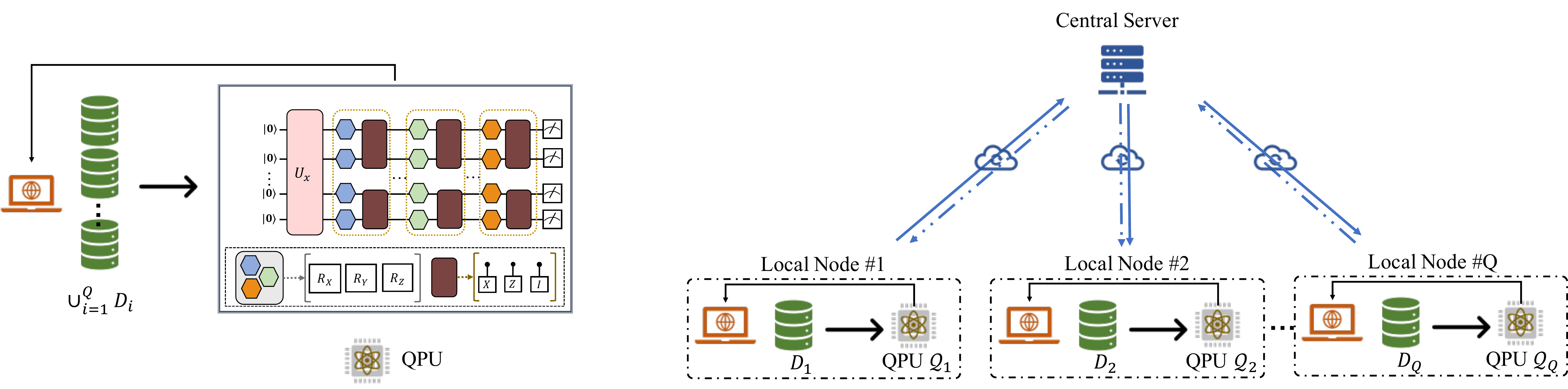}	
\caption{\small{\textbf{The scheme of VQAs and quantum distributed optimization scheme (QUDIO).} The left panel illustrates the workflow of conventional VQAs. The right panel presents the workflow of QUDIO, which consists of multiple local nodes and a central server. Each local node only manipulates a subgroup of the given problem in parallel, while the central server communicates with all local nodes and synchronizes  their results. In this way, QUDIO can accelerate various VQAs.    }}
\label{fig:sheme}
\end{figure*}

In light of the key role of distributed techniques in deep learning and the rapid growth of the amount of available quantum chips, a natural way to accelerate the training of VQAs is involving multiple quantum processors to fulfill the joint optimization. Noticeably,  designing such a scheme is extremely challenged and sharply inconsistent with classical distributed optimization methods \cite{boyd2011distributed,Jeffrey2012} or quantum distributed computation \cite{beals2013efficient}, because of the following three aspects. First, the gradients information operated in VQAs is biased, induced by the system imperfection such as sample error and gate noise, whereas most classical distributed methods assume the unbiased gradients information. Next, unlike distributed methods in DNNs, VQAs are immune to the communication bottleneck, due to the small number of trainable parameters. Last, quantum distributed computation   \cite{beals2013efficient} focuses on using multiple less powerful quantum circuits to simulate a standard quantum circuit, which differs from our aim. Although some quantum software \cite{bergholm2018pennylane,Qiskit} showcases parallelling variational quantum eigen-solvers (VQEs) \cite{peruzzo2014variational}, how to devise general distributed-VQAs optimization schemes with both runtime boost and convergence guarantee remains largely unknown.  

 To conquer the above issues, here we devise an efficient \texttt{QU}antum \texttt{DI}stributed \texttt{O}ptimization scheme (abbreviated as QUDIO). An attractive property of QUDIO is \textit{adequately} utilizing   the accessible quantum resources to accelerate VQAs, owing to its compatibility. Namely, the deployed quantum processors are allowed to be any type of quantum hardware such as linear optical, ion-trap, and superconducting quantum chips. Such a compatibility contributes to apply a wide class of VQAs to manipulate varied large-scale computational problems and seek potential quantum advantages by unifying  quantum powers in a maximum extent. 

Our second contribution is analyzing convergence of QUDIO under both the ideal and NISQ scenarios. Particularly, we prove an asymptotic convergence between QUDIO and conventional VQAs under the ideal setting. By contrast, in the NISQ case, the convergence rate of QUDIO becomes degraded with respect to the amplified system noise and the decreased number of quantum measurements. These results can not only be employed as theoretical guidance to assure good performance of QUDIO, but also motive us to devise more advanced distributed-VQAs schemes. To our best knowledge, this is the first proposal towards distributed-VQAs techniques with theoretical guarantee.  
 
Last, we conduct extensive numerical simulations to validate the computational efficiency of our scheme. In particular, QUDIO is exploited to accomplish image classification and the ground energy of hydrogen molecule estimation tasks under both the ideal and noisy scenarios. The achieved simulation results validate that QUDIO realizes sublinear and even superlinear speedups compared with conventional QNNs and VQEs.

The remainder of this study is organized as follows. In Section \ref{sec:QUDIO-all}, we present the workflow of QUDIO. Subsequently, in Section \ref{sec:QUDIO-QNN}, we exhibit how to use QUDIO to accelerate QNNs with both convergence analysis and numerical simulations. Next, we explain the usage of QUDIO to accelerate VQEs with numerical simulations in Section \ref{sec:QUDIO-VQE}. Last, in Section \ref{sec:conclusion}, we conclude this study and discuss future research directions.  

\section{Quantum distributed optimization scheme}\label{sec:QUDIO-all}

We depict the paradigm of QUDIO in the right panel of Fig.~\ref{fig:sheme} and present the corresponding Pseudocode in Alg.~\ref{alg:Dist-opt}. Conceptually, QUDIO is constituted by a classical central server and $Q$ local nodes $\{\mathcal{Q}_i\}_{i=1}^Q$, where each local node consists of a classical optimizer and a single quantum processor. The algorithmic implementation of QUDIO include three steps.
\begin{enumerate}
	\item At the preprocessing stage, the central server partitions the given problem into $Q$ subproblems and allocates each subproblem to each local node (see Sections \ref{sec:QUDIO-QNN} and \ref{sec:QUDIO-VQE} for concrete explanations). 
	\item The training procedure of QUDIO follows an iterative manner. Set the total number of global and local steps as $T$ and $W$, respectively. At the $t$-th global step with $t\in[T]$, the central server first dispatches the synchronized parameters $\bm{\theta}^{(t)}$ to $Q$ local nodes, which correspond to initial parameters for local updates (Line 5). With a slight abuse of notation, we denote $\bm{\theta}^{(t,w)}_i$ as the trainable parameters for the $i$-th local node at the $w$-th local step. The initial parameters for all local nodes satisfy
\[\bm{\theta}^{(t)}=\bm{\theta}^{(t,w=0)}_i,~ \forall i\in [Q].\] 
After initialization, all local nodes $\{\mathcal{Q}_i\}$ proceed $W$ iterations independently (highlighted by the red region in Lines 7-8). Once all local updates are fulfilled, the central server collects parameters $\{\bm{\theta}^{(t,W)}_i\}_{i=1}^Q$ from all local nodes to execute synchronization of trainable parameters (Line 11), i.e., 
\[\bm{\theta}^{(t+1)}=\frac{1}{Q}\sum_{i=1}^Q \bm{\theta}_i^{(t,W)},\] 
which completes the $t$-th global step. 
\item Through repeating the above procedure with $T$ global steps, the central server outputs the synchronized $\bm{\theta}^{(T)}$ as the optimized parameters. 
\end{enumerate}

\begin{algorithm}[H]
\caption{\small{\textbf{The Pseudocode of  QUDIO}. The codes highlighted by the yellow and pink shadows refer to execute them on the local nodes and the central server, respectively. }}
\label{alg:Dist-opt}
   \begin{algorithmic}[1]
   \State \textbf{Input}: The initialized parameters $\bm{\theta}^{(0)}\in [0,2\pi)^{d_Q}$, the employed loss function $\mathcal{L}$, the given dataset/Hamilton, the hyper-parameters $\{Q,\eta, W,T\}$  \;
   \State The central server partitions the given problem into $Q$ parts and allocates them to $Q$ local nodes \;
 \For{$t=0,\cdots,T-1$} \;
\For{Quantum processor $\mathcal{Q}_i$, $\forall i\in[Q]$ \textbf{in parallel}} \;
  \State $\bm{\theta}_{i}^{(t,0)}= \bm{\theta}^{(t)}$ \;
 \For{$w=0, \cdots, W-1$}
  \State \testbox[red]{Compute the estimate gradients $g_{i}^{(t,w)}$} \;
  \State \testbox[red]{Update  $\bm{\theta}_{i}^{(t, w + 1)}=\bm{\theta}_{i}^{(t,w)} - \eta g_{i}^{(t,w)}$} \;
 \EndFor
\EndFor
  \State \testbox[yellow]{Synchronize $\bm{\theta}^{(t + 1)} =   \frac{1}{Q}\sum_{i=1}^Q \bm{\theta}_{i}^{(t,W)}$}
\EndFor
\State \textbf{Output:} $\bm{\theta}^{(T)}$
\end{algorithmic}
\end{algorithm}

In principle, compared with original VQAs with single quantum processor, this parallel optimization mechanism enables QUDIO to \textit{reduce the computing time by a constant factor equal to the number of local nodes $Q$}. This linear speedup is warranted by the small amount of trainable parameters for most VQAs. Notably, this property differs the algorithmic design between distributed DNNs and distributed VQAs, where the former concerns a considerable communication bottleneck caused by billions of trainable parameters of DNNs. Furthermore, QUDIO is highly compatible and can be seamlessly embedded into cloud computing, since it supports various types of quantum processors to set up local nodes and its central server is purely classical.      

A core component in QUDIO is the approach of decomposing the given problem into $Q$ parts, which in turn results in the varied forms of the estimated gradients $\{g_i^{(t,w)}\}$ in Line 7 of Alg.~\ref{alg:Dist-opt}. For the purpose of elucidating, in the later context, we separately elaborate how to decompose the given problem and calculate the estimated gradients when applying QUDIO to speed up the training of QNNs and VQEs in Section \ref{sec:QUDIO-QNN} and Section \ref{sec:QUDIO-VQE}, respectively.

\textbf{Remark.} Although this study concentrates on QNNs and VQEs, our proposal can be \textit{effectively extended to speed up} other VQAs such as quantum approximate optimization algorithms \cite{amaro2021filtering,farhi2014quantum,hadfield2019quantum,zhang2021neural,zhou2020quantum}.

\section{Accelerate QNN by QUDIO}\label{sec:QUDIO-QNN}
Let us first formalize the classification task discussed in this section. Denote the given dataset as $\mathcal{D}=\{\bm{x}_i, y_i\}_{i=1}^n$, where $\bm{x}_i\in\mathbb{R}^{D_c}$ and $y_i\in\mathbb{R}$ refer to the features and label for the $i$-th example. The generic form of the output of QNNs \cite{du2020learnability,havlivcek2019supervised} is   
\begin{equation}\label{eqn:QNN-ansatz}
	h(\bm{\theta},O,\rho_i) = \Tr(OU(\bm{\theta})\rho_iU(\bm{\theta})^{\dagger}),~\forall i\in[n], 
\end{equation}
where $\rho_i$, $U(\bm{\theta})=\prod_{l=1}^LU_l(\bm{\theta})$, and $O \in \mathbb{C}^{2^N\times 2^N}$ are the encoded $N$-qubit state corresponding to $\bm{x}_i$, the ansatz with the circuit depth $L$ and $d_Q$ parameters (i.e., $\bm{\theta}\in [0,2\pi)^{d_Q}$), and the fixed quantum observable, respectively. Notably, the versatility of QNNs arises from diverse data embedding methods (e.g., preparing $\rho_i$ by basis, amplitude, and qubit encoding methods  \cite{larose2020robust}), flexible architectures of the ansatz $U(\bm{\theta})$ (e.g., hardware-efficient and tensor-network based ansatzes), and the agile choice of $O$. The aim of QNNs is to seek the optimal parameters $\bm{\theta}^*$ that minimize a predefined loss $\mathcal{L}$. Throughout the whole work, we specify $\mathcal{L}$ as the mean square error with $l_2$-norm regularizer, i.e., 
\begin{equation}\label{eqn:loss-QNN}
	\mathcal{L}(\bm{\theta}, \mathcal{D}) =    \frac{1}{2n}\sum_{i=1}^n  (h(\bm{\theta},O,\rho_i) -  y_i)^2 + \lambda\|\bm{\theta}\|^2_2,
\end{equation}
 where $\lambda \geq 0$ refers to the regularizer coefficient. The optimization of $\bm{\theta}$ can be accomplished by either using gradient-free or gradient-based methods. 
   
We now elaborate on how QUDIO in Alg.~\ref{alg:Dist-opt}   accelerates the training of QNN following Eq.~(\ref{eqn:loss-QNN}). Concretely, at the preprocessing stage, the central server splits the dataset $\mathcal{D}$ into $Q$ subgroups $\{\mathcal{D}_i\}_{i=1}^Q$ and assigns them into $Q$ local nodes $\{\mathcal{Q}_i\}_{i=1}^Q$. In the training procedure, QUDIO harnesses  the following iterative strategy to optimize the trainable parameters. At the $t$-th global step, when $Q$ local nodes receive the synchronized parameters $\bm{\theta}^{(t)}$ sent by the central server, they proceed $W$ local updates independently. Let $\bm{\theta}_i^{(t,w=0)}=\bm{\theta}^{(t)}$ for $\forall i\in[Q]$ and $\forall w\in[W]$. The updating rule of $\mathcal{Q}_i$ associated with the stochastic gradient descent optimizer    \cite{goodfellow2016deep} yields 
\begin{equation}\label{eqn:updatedrule-QNN}
	\bm{\theta}_i^{(t,w+1)} = \bm{\theta}_i^{(t,w)} - \eta   g_i(\bm{\theta}_i^{(t,w)}, \bm{x}_i^{(t,w)})\in [0,2\pi)^{d_Q},
\end{equation}       
where $\eta$ is the learning rate, the example $(\bm{x}_i^{(t,w)}, y^{(t)}_i)$ is uniformly sampled from $\mathcal{D}_i$, and $g_i(\bm{\theta}_i^{(t,w)}, \bm{x}_i^{(t,w)})$ denotes the estimation of $\nabla \mathcal{L}(\bm{\theta}_i^{(t,w)}, \bm{x}_i^{(t,w)})$  induced by the system noise and sample error. Once all local updates are completed, the central server receives  parameters $\{\bm{\theta}_i^{(t,W)}\}_{i=1}^Q$ and synchronizes them to update the global trainable parameters, i.e.,
\begin{equation}\label{eqn:QUDIO-glb-upd}
	\bm{\theta}^{(t+1)} = \frac{1}{Q} \sum_{i=1}^Q \bm{\theta}^{(t,W)}_i.
\end{equation}   
Through repeating the above process with $T$ times, QUDIO outputs $\bm{\theta}^{(T)}$ as the trained parameters.   

In the remainder of this section, we first explain the acquisition of the estimated gradients $g_{i}^{(t,w)}$ and analyze the convergence of QUDIO. We then benchmark  performance of QUDIO towards image classification tasks. 

\subsection{The acquisition of the estimated gradients }
The explicit form of $ g_i(\bm{\theta}_i^{(t,w)}, \bm{x}_i^{(t,w)})$ in Eq.~(\ref{eqn:updatedrule-QNN}) is established on the analytic gradient $\nabla \mathcal{L}_i(\bm{\theta}_i^{(t,w)}, \bm{x}_i^{(t,w)})$. In this perspective, here we first recap the mathematical expression of $\nabla \mathcal{L}_i(\bm{\theta}_i^{(t,w)}, \bm{x}_i^{(t,w)})$. According to Eq.~(\ref{eqn:updatedrule-QNN}), the updating rule of $\mathcal{Q}_i$ in the \textit{ideal scenario} yields 
\begin{equation}\label{eq:param_update}
	\bm{\theta}_i^{(t,w+1)} = \bm{\theta}_i^{(t,w)} - \eta \nabla \mathcal{L}_i(\bm{\theta}_i^{(t,w)}, \bm{x}_i^{(t,w)}),
\end{equation}  
where $\mathcal{L}_i(\bm{\theta}^{(t,w)}_i, \bm{x}_i^{(t,w)}) =  \frac{1}{2} (\hat{y}_i^{(t,w)}-  y_i^{(t)})^2 + \lambda\|\bm{\theta}^{(t,w)}_i\|^2_2$ and $\hat{y}_i^{(t,w)} = h(\bm{\theta}^{(t,w)}_i,O,\rho_i^{(t)}) $ refers to the prediction of QNN for the sampled example $\bm{x}_i^{(t,w)}$ as defined in  Eq.~(\ref{eqn:QNN-ansatz}). The evaluation of the analytic gradient can be achieved via the parameter shift rule \cite{mitarai2018quantum,schuld2019evaluating}, i.e., the $j$-th component of $\nabla_j \mathcal{L}_i(\bm{\theta}_i^{(t,w)}, \bm{x}_i^{(t,w)})$ for $\forall j\in[d_Q]$ satisfies \begin{equation}\label{eqn:analy-grad-QNN}
   (\hat{y}_i^{(t,w)} - y_i^{(t)}) \frac{\hat{y}_i^{(t,w,+_j)}- \hat{y}_i^{(t,w,-_j)}}{2}+\lambda \bm{\theta}_{i,j}^{(t,w)},
\end{equation} 
where $\hat{y}_i^{(t,w,\pm_j)}=h(\bm{\theta}^{(t,w,\pm)}_i,O,\rho_i^{(t)})$ denotes the outputs of QNN with shifted parameters $\bm{\theta}^{(t,w,\pm)}_i=\bm{\theta}^{(t,w)}_i\pm \frac{\pi}{2} \bm{e}_j$.   

In the NISQ scenario, the system noise and the sample error forbid the acquisition of the analytic gradients. Instead, the classical optimizer can only collect the estimated gradients. More precisely, suppose that the depolarization noise channel $\mathcal{N}_p(\cdot)$ is injected to each quantum circuit depth, i.e.,
\begin{equation}\label{eqn:dep_noise}
	\mathcal{N}_p(\rho)=(1-p)\rho+p\frac{\mathbb{I}}{2^N}.
\end{equation}
The output state before measurements is 
\[\gamma_i^{(t,w)}=(1-\tilde{p})U(\bm{\theta}^{(t,w)})\rho_i^{(t)}U(\bm{\theta}^{(t,w)})^{\dagger}+\tilde{p}\frac{\mathbb{I}}{2^N},\] 
where $\tilde{p}=1-(1-p)^{L_Q}$ and $L_Q$ refers to the total circuit depth \cite{du2020learnability}.     Suppose $\{O, \mathbb{I}- O\}$ in Eq.~(\ref{eqn:QNN-ansatz}) refers to a two-outcome positive operator valued measure (POVM) \cite{nielsen2010quantum}. Then a quantum measurement on the state $\gamma_i^{(t,w)}$ produces the outcome that can be viewed as a binary random variable with the Bernoulli distribution, i.e., $V_k^{(t,w)}\sim\Ber(\Tr(O\gamma_i^{(t,w)}))$ . In this way, the sample mean $\bar{y}_i^{(t,w)}=\sum_{k=1}^K V_k^{(t,w)}/K$ corresponding to $\Tr(O\gamma_i^{(t,w)})$ is obtained after $K$ measurements. In conjunction with above observations and the analytic gradients in Eq.~(\ref{eqn:analy-grad-QNN}), the explicit form of the estimated gradients for the $j$-th component $g_{i,j}(\bm{\theta}_i^{(t,w)}, \bm{x}_i^{(t,w)})$ with $\forall j\in[d_Q]$ satisfies    
\begin{equation}\label{eqn:est-grad-QNN}
 (\bar{y}_i^{(t,w)} - y_i^{(t)}) \frac{\bar{y}_i^{(t,w,+_j)}- \bar{y}_i^{(t,w,-_j)}}{2}+\lambda \bm{\theta}_{i,j}^{(t,w)},
\end{equation} 
where $\bar{y}_i^{(t,w,\pm_j)}=\frac{1}{K}\sum_{k=1}^K V_k^{(\pm_j)}$ are  estimated outputs with shifted parameters, i.e., $V_k^{(\pm_j)}\sim \Ber(\Tr(O\gamma_i^{(\pm_j)}))$ and  $\gamma_i^{(\pm_j)}= (1-\tilde{p})U(\bm{\theta}^{(t,w,\pm_j)})\rho_iU(\bm{\theta}^{(t,w, \pm_j)})^{\dagger}+\tilde{p}\frac{\mathbb{I}}{2^N}$.

To better understand the capability of QUDIO, we further analyze its  convergence rate. The convergence is quantified by the utility  $R_1 = \frac{1}{T}\sum_{t=1}^T \mathbb{E}[\|\nabla \mathcal{L}(\bm{\theta}^{(t)})\|^2]$, where the expectation is taken over the randomness of data sample, the imperfection of quantum system, and the finite quantum measurements. Intuitively, the metric $R_1$ evaluates how far QNN is away to the stationary points, which is a standard measure in non-convex optimization theory \cite{jain2017non,sun2019optimization}. The following theorem summarizes the convergence rate of QUDIO, whose proof is provided in Appendix \ref{append:QUDIO-QNN-conv}.  
\begin{theorem}\label{thm:conv-qnn}
Assume that the discrepancy between the collected local gradients and the analytic gradients is bounded, i.e., $\forall i\in[Q]$ and  $\forall \bm{\theta}\in[0,2\pi)^{d_Q}$, there exists
	\begin{equation}\label{assum:bound_vari}
		\mathbb{E}_{\bm{x}_i^{(t,w)} \sim \mathcal{D}_i}\left[ \left\|    g_i^{(t,w)} - \nabla \barL_i \left(\bm{\theta}_i^{(t,w)},  \bm{x}_i^{(t,w)}  \right) \right\|^2\right] \leq \sigma, .
	\end{equation}
Following notations in Eqs.~(\ref{eqn:QNN-ansatz})-(\ref{eqn:est-grad-QNN}),  when the system noise is modeled by the depolarization channel in Eq.~(\ref{eqn:dep_noise}) and the number of measurements to estimate the expectation value is $K$, under a mild assumption,  the convergence of QUDIO yields
	\begin{equation}
		R_1\leq O \left( \lambda d_Q\sqrt{\frac{S}{T}} +       \sqrt{\frac{S}{T}}  \left(4W^2 \sigma^2 + 2W^2 G_2^2 \right)     + C_1\right), \nonumber
	\end{equation} 
where    
$C_1\sim O(W^2(\tilde{p}d_Q+ \tilde{p}d_Q/K ))$, $\tilde{p}=(1-(1-p)^{L_Q})$, and $L_Q$ is the total circuit depth.
\end{theorem}
    
\begin{figure*}
\captionsetup[subfigure]{justification=centering}
\centering
\includegraphics[width=0.98\textwidth]{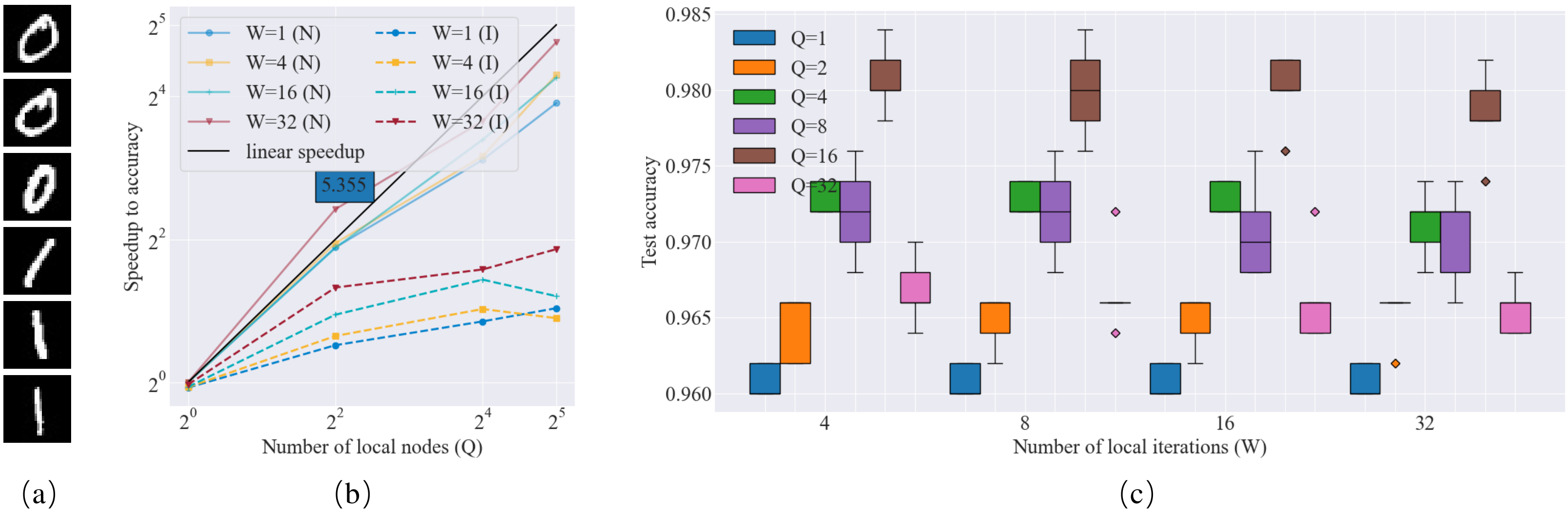}
\caption{\small{\textbf{Simulation results of QUDIO towards hand-written digits image classification.} (a) A visualization of some training examples sampled from the MNIST dataset. (b) Scaling behavior of QUDIO in clock-time for increasing number of local nodes $Q$ for varied number of local steps $W$. The labels `$W=a$  (I)' and `$W=a$  (N)' refer that the total number of local iterations is $W=a$ under the ideal and NISQ scenarios respectively. The hyper-parameters settings for the NISQ case are $p=10^{-5}$ and $K=100$. (c) A box plot that illustrates the achieved test accuracy of QUDIO with varied $W$ and $Q$ in the NISQ scenario, where the hyper-parameters settings are same with those described in (b).  }}
\label{fig:qnn}
\end{figure*}

The results of Theorem \ref{thm:conv-qnn} deliver three-fold implications. First, large system noise and few number of measurements may induce the optimization of QUDIO to be divergent, since the term $C_1$ is independent with $T$ and is amplified by $p$ and $1/K$. This observation hints the importance of integrating error mitigation techniques \cite{cai2020mitigating,du2020quantum,mcclean2020decoding,strikis2020learning} into QUDIO to enhance its trainability. Second, in the NISQ scenario, reducing the iteration number of local updating $W$ suggests a better performance, because $C_1$ is proportional to $W$. This phenomenon is starkly contrast with classical distributed optimization methods, which adopt large $W$ to alleviate the communication overhead. Last but not least, under the ideal setting, the convergence rate between conventional QNNs and QUDIO is identical, i.e., both of them scale with $O(1/\sqrt{T})$ with respect to the step number $T$ \cite{du2020learnability}. Celebrated by the joint optimization strategy, the similar convergence rate warrants that QUDIO promises a linear runtime speedup with respect to the increased number of local nodes $Q$.  
 
\textbf{Remark.} The developed tools in the proof of Theorem \ref{thm:conv-qnn} can be extended to analyze QNNs with other loss functions, quantum noisy models, and optimizers. Moreover, we would like to emphasize that naively imitating classical distributed algorithms to design distributed VQAs is suboptimal, since the inevitable biased gradient information in the quantum scenario may incur a deficient convergence.

\subsection{Numerical simulations}
We carry out numerical simulations to exhibit how QUDIO accelerates QNNs when dealing with a standard binary classification task with a large size of training examples. More precisely, the exploited dataset is distilled from a hand-written digits images dataset, called MNIST dataset \cite{lecun1998mnist}, which contains $256$ training examples and $500$ test examples labeled with digits `0' and `1'. Fig.~\ref{fig:qnn}(a) visualizes some examples in the distilled dataset. The amplitude encoding method and the  hardware-efficient ansatz are used to set up all $Q$ local nodes. The hyper-parameters settings are as follows. The number $Q$ ranges from $1$ to $32$. The number of local iterations has six settings, i.e., $W\in[1,2,4,8,16,32]$. In the NISQ setting, we set $K\in[5,100]$ and $p\in[10^{-4}, 10^{-1}]$. Each setting is repeated with $5$ times to collect the statistical results. See Appendix \ref{app:qnn} for the omitted implementation details.

To better quantify the performance of QUDIO from different angles, we introduce two metrics, i.e., the speedup to accuracy and the test accuracy, to evaluate the achieved results. Namely, the former considers the speedup ratio of QUDIO compared with the setting $Q=1$, i.e., supposes that the train accuracy reaches a predefined threshold (e.g., $95\%$) in  $T_1$ ($T_2$) clock-time for $Q=1$ ($Q=a$), the speedup to accuracy is evaluated by $T1/T2$. The latter allows us to compare the top test accuracy of QUDIO within a fixed number of global steps $T$ with varied $Q$ and $W$.

Fig.~\ref{fig:qnn} exhibits our simulation results. As shown in Fig.~\ref{fig:qnn} (b), for both the ideal and NISQ cases, QUDIO gains the speedup when increasing the number of local nodes $Q$. Strikingly, QUDIO can even reach a superlinear speedup in the NISQ scenario when $W=32$, e.g., it achieves $5.355$ times speedup for $Q=4$. This phenomenon indicates that QUDIO is insensitive to the communication bottleneck, which differs from distributed-DNNs \cite{Jeffrey2012}. Moreover, the distinct scaling behavior of QUDIO between the ideal and the NISQ cases is mainly caused by the fact that the evaluation of the analytic gradients in the ideal case is extremely fast and the communication cost dominates the runtime cost.  Fig.~\ref{fig:qnn}(c) pictures the statistical results of test accuracy for QUDIO. For all settings of $Q$, an increased $W$ generally degrades the performance of QUDIO in the statistical view. These results partially echo with Theorem \ref{thm:conv-qnn} such that larger $W$ suggests worse performance. An evidence is when $Q=16$, QUDIO achieves the best test accuracy in $W=4$. See Appendix \ref{app:qnn} for more  results and comprehensive investigation about the capability of QUDIO.

\begin{figure*}[htp]
\captionsetup[subfigure]{justification=centering}
\centering
\includegraphics[width=0.98\textwidth]{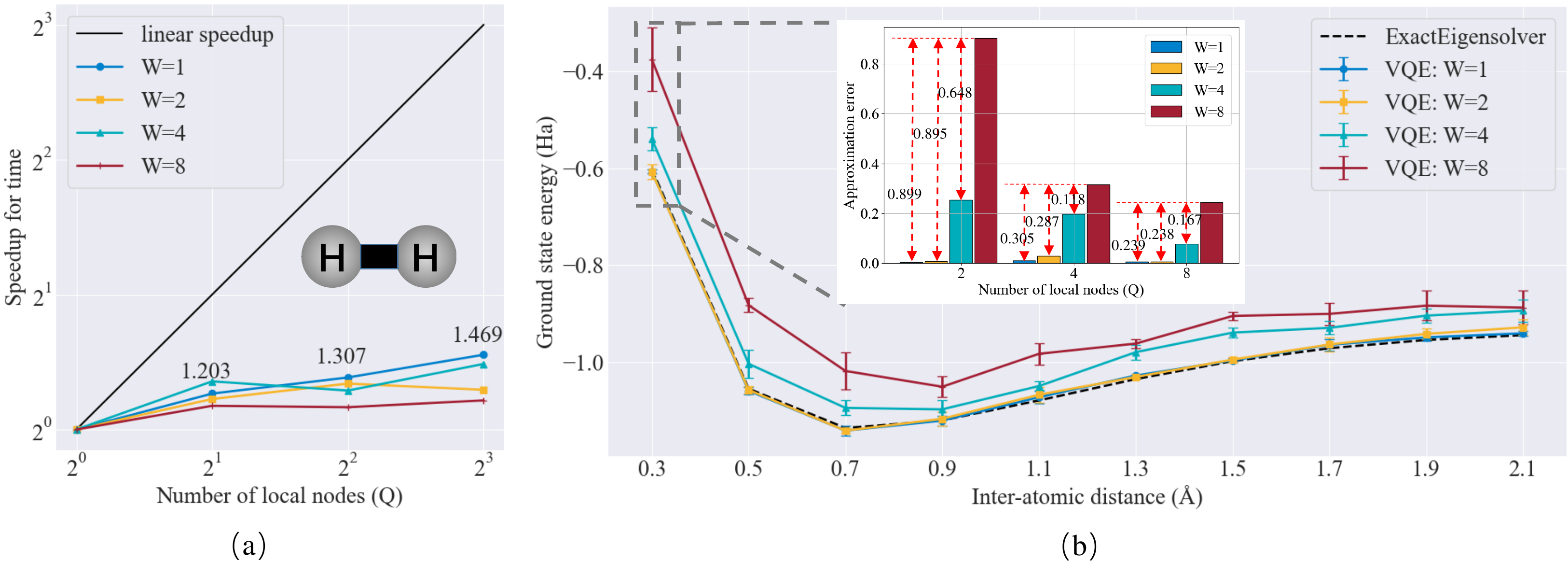}
\caption{\small{\textbf{Simulation results of QUDIO towards the ground state energy estimation of hydrogen molecule.} (a) The  speedup ratio with respect to different number of local nodes $Q$ and local iterations $W$. (b) The potential energy surface estimated by QUDIO. The black dotted line represents the exact ground state energy. The inner plot compares the  error between the ground truth and the estimated results of QUDIO with the case of $0.3\mathrm{\AA}$ inter-atomic distance.}}
\label{fig:vqe}
\end{figure*}

\section{Accelerate VQE by QUDIO}\label{sec:QUDIO-VQE} Variational quantum eigen-solvers (VQEs) \cite{Cervera2021meta-variational,kandala2017hardware,peruzzo2014variational,tang2021qubit} belong to another pivotal subclass of VQAs and have a broad usage of tackling quantum chemistry problems such as ground state estimation. The paradigm of VQEs is analogous to QNNs and other VQAs, which is completed by minimizing a problem-specific loss via gradient descent methods. Define the input Hamiltonian as 
\[H=\sum_{i=1}^n  \alpha_i H_i \in \mathbb{C}^{2^N\times 2^N},\] where $H_i$ refers to the $i$-th local Hamiltonian term and $\alpha_i\in\mathbb{R}$ is the corresponding coefficient. Without loss of generality, suppose that $H_i\in\{\X,\Y,\Z,\I\}^{\otimes N}$ is generated by Pauli operators.  The loss function of VQEs yields
\begin{equation}\label{eqn:loss_VQE}
	\mathcal{L}(\bm{\theta}, H) =    \Tr(H U(\bm{\theta})\rho_0 U(\bm{\theta})^{\dagger}),
\end{equation}
where $\rho_0=(\ket{0}\bra{0})^{\otimes N}$ is a fixed $N$-qubit quantum state and  $U(\bm{\theta})$ is the ansatz defined in Eq.~(\ref{eqn:QNN-ansatz}). Note that to evaluate $\mathcal{L}(\bm{\theta}, H)$, conventional VQEs use a single quantum chip to sequentially compute the results $\{\Tr(H_i U(\bm{\theta})\rho_0 U(\bm{\theta})^{\dagger})\}_{i=1}^n$ followed by a linear combination with $\{\alpha_i\}$. This implies the computational hardness of VQEs when the number of terms $n$ becomes large.         
 
When applying QUDIO to accelerate VQEs, the central server  splits the set of local Hamiltonians and their coefficients into $Q$ subgroups $\{\mathcal{S}_i\}_{i=1}^Q$ and assigns them into $Q$ local nodes $\{\mathcal{Q}_i\}_{i=1}^Q$. For instance, in the extreme case of $Q=n$, we have $\mathcal{S}_i= \{(\alpha_i, H_i)\}$. For the case of $Q<n$, $n$ local hamiltonians are divided into $Q$ subgroups such that $\mathcal{S}_i=\cup_{j\in S_i} \{(\alpha_j, H_j)\}$, where $S_i$ refers to the $i$-th index set with $\cup_{i=1}^Q S_i =[n]$ and $S_i\cap S_j=\emptyset$ when $i\neq j$. Moreover, the training procedure follows the same manner with accelerating QNNs.  The updating rule of the local node $\mathcal{Q}_i$ for $\forall i\in[Q]$ yields
\begin{equation}
    \bm{\theta}_i^{(t,w+1)} = \bm{\theta}_i^{(t,w)} - \eta   g_i(\bm{\theta}_i^{(t,w)}, H_{S_i}),
\end{equation}
where $g_i(\bm{\theta}_i^{(t,w)}, H_{S_i})$ refers to the estimated gradient of $\nabla \mathcal{L}_i(\bm{\theta}_i^{(t,w)}, H_{s_i})=\Tr(H_{S_i}U(\bm{\theta}_i^{(t,w)})\rho_0U(\bm{\theta}_i^{(t,w)})^\dagger)$ with $H_{S_i}=\sum_{j\in S_i}\alpha_j H_j$.   
  
In the subsequent subsections, we first exhibit the explicit form of the estimated gradient $g_i(\bm{\theta}_i^{(t,w)}, H_{S_i})$ and then conduct extensive numerical simulations to validate performance of QUDIO.  Remarkably, to facilitate simulation, here we mainly focus on the scenario in which the estimation error is caused by the finite number of measurement and the system is noiseless.

\subsection{The acquisition of the estimated gradients }
Before diving into deriving the explicit of estimated gradient, let us first recall the analytic gradient of $\mathcal{L}_i(\bm{\theta}_i^{(t,w)}, H_{s_i})$. Specifically, based on the  parameter shift rule, the $j$-th component of the analytic  gradients is 
\begin{equation}\label{eqn:grad-analy-VQE}
 \nabla_j\mathcal{L}_i(\bm{\theta}_i^{(t,w)}, H_{S_i})=   \frac{\hat{y}_i^{(t,w,+_j)}- \hat{y}_i^{(t,w,-_j)}}{2},
\end{equation}
where $\bm{\theta}^{(t,w,\pm)}_i=\bm{\theta}^{(t,w)}_i\pm \frac{\pi}{2} \bm{e}_j$ and  $\hat{y}_i^{(t,w,\pm_j)}=\sum_{j\in S_i}\alpha_j \Tr(H_{j}U(\bm{\theta}_i^{(t,w,\pm)})\rho_0U(\bm{\theta}_i^{(t,w,\pm)})^\dagger)$.

When a finite number of measurement is allowable, the trace terms of $\hat{y}_i^{(t,w,\pm_j)}$ in Eq.~(\ref{eqn:grad-analy-VQE}) can only be acquired with estimation error. The detailed procedure to estimate $\Tr(H_{j}U(\bm{\theta}_i^{(t,w,\pm)})\rho_0U(\bm{\theta}_i^{(t,w,\pm)})^\dagger)$ is as follows. To estimate this result by measuring the quantum state $\rho=U(\bm{\theta}_i^{(t,w)})\rho_0U(\bm{\theta}_i^{(t,w)})^\dagger$ along the computational basis, an alignment operation should be executed. Mathematically, the quantum state $\rho$ needs to interact with the unitary operator $R_i$ to generate the quantum state $\rho'=R_j\rho R_j^\dagger$, where $R_j$ is composed of a sequence of rotational single-qubit gates whose row vectors are the eigen-basis of $H_j$. Define $\bm{h}_i$ as a vector that collects the eigenvalues of $H_j$,  $V_k\sim \mathrm{Cat}(2^N, \bm{p}(\rho'))$ as a random variable following the categorical distribution (also called generalized Bernoulli distribution), and  $\bm{p}(\rho')\in\mathbb{R}^{2^N}$ refers as a discrete distribution with $\bm{p}_j(\rho')=\mathrm{Tr}(\rho'\ket{j}\bra{j})$. Following the above notations, the term $\Tr(H_j\rho')$ is estimated by $\frac{1}{K}\sum_{k=1}^K\bm{h}_{i,V_k}$, where $\bm{h}_{i,V_k}$ is the $V_k$-th eigenvalue of $H_j$. Based on Eq.~(\ref{eqn:grad-analy-VQE}), the estimated gradient satisfies 
\begin{equation}
 g_{i,j}(\bm{\theta}_i^{(t,w)}, H_{S_i}) = \frac{\bar{y}_i^{(t,w,+_j)}- \bar{y}_i^{(t,w,-_j)}}{2}, \forall j \in[d_Q],
\end{equation}
where $\bar{y}_i^{(t,w,\pm_j)}=\sum_{i\in S_i}\alpha_i\frac{1}{K}\sum_{k=1}^K\bm{h}_{i,V_k^{(\pm_j)}}$  and the definition of $V_k^{(\pm_j)}\sim \mathrm{Cat}(2^N, \bm{p}(\rho^{(\pm_j)}))$ follows the same manner with Eq.~(\ref{eqn:est-grad-QNN}).

Considering that the training procedure is exactly identical to the way of applying QUDIO to accelerate QNNs, Theorem \ref{thm:conv-qnn} can also describe the convergence behavior of QUDIO for  accelerating VQEs. 

\subsection{Numerical simulations}          
We perform numerical simulations to validate the effectiveness of QUDIO to accelerate conventional VQEs. To do so, we apply QUDIO to estimate the ground state energy of hydrogen molecule $\mathrm{H}_2$ with varied bond distance, whose Hamiltonian contains $n=15$ local Hamiltonian terms and requires $N=4$ qubits with $H=\sum_{i=1}^{15}\alpha_i H_i\in \mathbb{C}^{2^4\times 2^4}$ \cite{kandala2017hardware}. The implementation of all local nodes mainly follows the proposal \cite{kandala2017hardware} such that the input quantum state is $\rho_0=\ket{1100}\bra{1100}$ and the trainable unitary refers to the hardware-efficient ansatz. The hyper-parameters settings are as follows. The number of local nodes and local iterations is set as $Q\in\{1, 2, 4 ,8\}$ and  $W\in[1,2,4,8]$, respectively. We fix the number of measurements to be $K=100$. Each setting is repeated with $5$ times to collect the statistical results. See Appendix \ref{app:vqe} for the omitted implementation details.

The simulation results are shown in  Fig.~\ref{fig:vqe}. In particular, the left subplot illustrates the speedup ratio of QUDIO in terms of the factors $Q$ and $W$. For all settings of $W$, QUDIO gains runtime speedups by involving more local nodes. For example, we obtain $1.469$ times acceleration when utilizing $8$ local nodes to optimize VQE with $1$ local update. Although QUDIO provides certain speedups in the task of estimating the ground state energy of hydrogen molecule, there is a clear gap towards the linear speedup ratio. It is noteworthy that this gap arises from the simplicity of the manipulated problem, where the communication overhead dominates the total computational runtime. We expect that QUDIO has the ability to earn higher speedup ratio for large-scale tasks.

The potential energy surface estimated by QUDIO is presented in Fig.~\ref{fig:vqe}(b). The outer plot suggests that for all bond distance settings ranging from $0.3\mathrm{\AA}$ to $2.1\mathrm{\AA}$, QUDIO obtains the best performance with $W=1$, which is almost the same with the exact values. By contrast, there exists an apparent separation between the exact values and the estimated results of QUDIO with $W=8$. The inner plot further evidences this phenomenon. Specifically, when the bond distance equals to  $0.3\mathrm{\AA}$, QUDIO witnesses the largest approximation error $0.9Ha$ with $W=8$ and $Q=2$, while the error is reduced to nearly zero with $W=1$ regardless of the number of local nodes. All of the above observations collaborate with Theorem \ref{thm:conv-qnn}, where decreasing $W$ warrants a better performance. Refer to Appendix \ref{app:vqe} for deep comprehension.

\section{Discussion and conclusion}\label{sec:conclusion}
\indent In this study, we devise QUDIO to accelerate VQAs with multiple quantum processors. We also provide theoretical analysis about how the system noise and the number of measurements influence the convergence of QUDIO. An attractive feature is that in the ideal setting, QUDIO obeys the asymptotic convergence rate with conventional QNNs, which ensures its runtime speedup with respect to the increased number of local nodes. The achieved numerical simulation results confirm the effectiveness of our proposal.  Particularly, in the NISQ scenario, QUDIO can achieve superline speedups in the measure of time-to-accuracy.

We remark that there are three orthogonal research directions towards the investigation of distributed VQAs. First,  instead of employing the synchronization approach used in QUDIO, it is intrigued to design asynchronous distributed-VQAs schemes with convergence guarantees, which may further reduce the communication overhead and maximally utilize quantum processors with distinct  qualities. Second, with the aim of reducing the runtime cost, it is important to integrate the effective measurement reduction algorithms with QUDIO and other distributed VQAs. Repressive examples contain grouping compatible operators \cite{kandala2017hardware,zhao2020measurement} and classical-shadows based methods \cite{huang2020predicting,struchalin2021experimental}. Last, a promising direction is combining QUDIO with a recent work \cite{zhang2021variational}, which splits the input quantum circuits into several individual quantum circuits with distributed optimization.               

The employment of cloud computing to execute QUDIO naturally invokes the security issue \cite{lu2020quantum}. For example, how to defend adversarial attack or prevent the private information leakage        during the training procedure. Initial studies have leveraged some notions such as differential privacy \cite{du2020quantum,du2021quantum}, hypothesis testing \cite{weber2021optimal}, and blind quantum computing \cite{li2021quantum} to address this issue. However, little is known about how these strategies effect the convergence rate. A deep understanding towards this topic is highly desired, which determines the applicability of distributed VQAs.

For these reasons, QUDIO and its variants, which marry the distributed techniques with VQAs, could substantially contribute to use NISQ machines to accomplish real-world problems with quantum advantages.


 \newpage   
\clearpage 
\medskip

\newpage   
\clearpage 
\appendix 
\onecolumngrid

\section{The proof of Theorem \ref{thm:conv-qnn}}\label{append:QUDIO-QNN-conv}

The outline of this section is as follows. In Appendix \ref{subsec:AppendB-1}, we  simplify some notations and introduce basic concepts in optimization theory for ease of discussion. Next, in Appendix \ref{subsec:AppendB-2},  we demonstrate the proof details of Theorem \ref{thm:conv-qnn}.  

\subsection{Basic notations and concepts}\label{subsec:AppendB-1}

\noindent\textbf{Notations.} Let us first simplify some notations defined in the main text to facilitate the derivation. Recall that the loss function of the $i$-th local node is denoted by $\mathcal{L}_i\left(\bm{\theta}_i^{(t,w)}, \bm{x}_i^{(t,w)}\right)$ with $\bm{x}_i^{(t,w)}$ being an example uniformly sampled from the sub-dataset $\mathcal{D}_i$. In this section, when no confusion occurs, we define the loss for the $i$-th local node $\forall i\in[Q]$ with the \textit{whole} sub-dataset $\mathcal{D}_i$ as 
\begin{equation}
	\mathcal{L}_i\left(\bm{\theta}_i^{(t,w)}, \mathcal{D}_i\right) \equiv  \mathcal{L}_i\left(\bm{\theta}_i^{(t,w)}\right)= \frac{1}{2|\mathcal{D}_i|}\sum_{j=1}^{|\mathcal{D}_i|}  \left(h(\bm{\theta}_i^{(t,w)},O,\rho_j) -  y_j\right)^2 + \lambda\|\bm{\theta}_i^{(t,w)}\|^2_2, 
\end{equation}      
where $\rho_j$ and $y_j$ refers to the encoded quantum example and corresponding label with respect to the $j$-th example in $\mathcal{D}_i$. Besides, we simplify the global loss function as   
\begin{equation}
	\mathcal{L}\left(\bm{\theta}_i^{(t,w)}, \mathcal{D}_i\right) \equiv  \mathcal{L}\left(\bm{\theta}_i^{(t,w)}\right) = \frac{1}{Q}\sum_{i=1}^Q \mathcal{L}_i \left(\bm{\theta}_i^{(t,w)}, \mathcal{D}_i\right).
\end{equation}
Following the same routine,  the gradients of the loss $\mathcal{L}_i(\bm{\theta}_i^{(t,w)})$ (or $\mathcal{L}\left(\bm{\theta}_i^{(t,w)}\right)$) are written as $\nabla \mathcal{L}_i(\bm{\theta}_i^{(t,w)}, \mathcal{D}_i)$ and $ \nabla \mathcal{L}_i(\bm{\theta}_i^{(t,w)})$ (or $\nabla \mathcal{L}(\bm{\theta}_i^{(t,w)}, \mathcal{D}_i)$ and $ \nabla \mathcal{L}(\bm{\theta}_i^{(t,w)})$) interchangeably. 

In QUDIO, to decrease runtime,  the optimizer of the $i$-th local node for $\forall i\in[Q]$ only requires the gradients of $\mathcal{L}_i\left(\bm{\theta}_i^{(t,w)}, \bm{x}_i^{(t,w)}\right)$ instead of the whole sub-dataset $\mathcal{D}_i$. To distinguish with $\mathcal{L}_i(\bm{\theta}_i^{(t,w)})$, we denote 
\begin{equation}
	\mathcal{L}_i\left(\bm{\theta}_i^{(t,w)},  \bm{x}_i^{(t,w)}\right) = \barL_i\left(\bm{\theta}_i^{(t,w)}\right),~\text{and}~\nabla \mathcal{L}_i\left(\bm{\theta}_i^{(t,w)},  \bm{x}_i^{(t,w)}\right) = \nabla \barL_i\left(\bm{\theta}_i^{(t,w)}\right).
\end{equation}     

We note that the gradients between $ \nabla \mathcal{L}_i\left(\bm{\theta}_i^{(t,w)}\right)$ and $\nabla \barL_i\left(\bm{\theta}_i^{(t,w)}\right)$ have the relationship 
\begin{equation}\label{append:eqn:eqiv-est-analy}
	\mathbb{E}_{\bm{x}_i^{(t,w)}\sim \mathcal{D}_i} \left[ \nabla \barL_i\left(\bm{\theta}_i^{(t,w)}\right) \right] = \nabla \mathcal{L}_i\left(\bm{\theta}_i^{(t,w)}\right),
\end{equation}
where the expectation is taken over the randomness of sampled examples. 

\medskip 
\noindent\textbf{Basic concepts in optimization theory.} We introduce two definitions, i.e., $S$-smooth and $G$-Lipschitz \cite{boyd2004convex}, which are employed to quantify properties of loss functions and achieve the proof of Theorem \ref{thm:conv-qnn}.
\begin{definition}\label{def:S-smoo-G_lip}
A function f is $S$-smooth over a set $\mathcal{C}$ if $\nabla^2 f(\bm{u})\preceq S\mathbb{I}$ with $S>0$ and $\forall \bm{u}\in \mathcal{C}$. A function f is  $G$-Lipschitz over a set $\mathcal{C}$ if for all $\bm{u},\bm{w}\in \mathcal{C}$, we have $|f(\bm{u}) - f(\bm{w})|\leq G\|\bm{u}-\bm{w}\|_2$.  
\end{definition}

As proved in the study \cite{du2020learnability}, the mean square error loss in Eq.~(\ref{eqn:loss-QNN}) is smooth and Lipschitz. 
\begin{lemma}[Lemma 2, \cite{du2020learnability}]\label{lem:Lsmooth}
 The loss function $\mathcal{L}$ in Eq.~(\ref{eqn:loss-QNN})  is $S$-smooth with $S = (3/2+ \lambda)d^2$ and $G_1$-Lipschitz with $G_1=d(1+3\pi\lambda)$.  
 \end{lemma}

\subsection{Proof details}\label{subsec:AppendB-2}
The proof of Theorem \ref{thm:conv-qnn} exploits the relation between the analytic and estimated gradients of QNN, i.e., $\nabla_j\barL_i(\bm{\theta}_i^{(t,w)})$ and $\nabla_j\mathcal{L}_i(\bm{\theta}_i^{(t,w)})$. 
\begin{lemma}\label{lem:thm1}
Denote $\tilde{p}=1-(1-p)^{L_Q}$ with $L_Q$ being the quantum circuit depth.	The discrepancy between the analytic gradients $\nabla \mathcal{L}_i(\bm{\theta}_i^{(t,w)})$ and the estimated gradients $\nabla \barL_i(\bm{\theta}_i^{(t,w)})$ for the $i$-th node satisfies
	\begin{eqnarray}
		&& \mathbb{E}_{ \bm{\varsigma}_{i,j}^{(t,w)}}\left[ \left\|    \nabla \barL_i(\bm{\theta}_i^{(t,w)}) - \nabla \mathcal{L}_i\left(\bm{\theta}_i^{(t,w)} \right) \right\|^2\right] \nonumber\\
		\leq && (\tilde{p}-2)^2\tilde{p}^2 \left\|\nabla \mathcal{L}_i\left(\bm{\theta}_i^{(t,w)}\right) \right\|^2 + \frac{(1-\tilde{p})^2\tilde{p}^2d_Q}{4} + (2-\tilde{p})^2\tilde{p}^2 G_1d_Q + \frac{7(1 - \frac{\tilde{p}}{2})^2 + \frac{1}{8}}{K}d_Q,
	\end{eqnarray}
	where the expectation is taking over the randomness of quantum system noise and the measurement error.
\end{lemma}
The proof of the above lemma is given in Appendix \ref{append:subsec:proof-lem2}.

We are now ready to prove Theorem \ref{thm:conv-qnn}. 
\begin{proof}[Proof of Theorem \ref{thm:conv-qnn}]
 
The machinery of QUDIO in Alg.~\ref{alg:Dist-opt}  indicates that the difference between the trainable parameters at the $(t+1)$-th and $t$-th global steps   satisfies 
\begin{equation}\label{eqn:thm_emr_QNN_0}
	\bm{\theta}^{(t + 1)} -  \bm{\theta}^{(t)} = - \frac{\eta}{Q} \sum_{i=1}^Q\sum_{w=1}^W g_i^{(t,w)} ~.
\end{equation}

Supported by $S$-smooth property of  $\mathcal{L}(\bm{\theta})$ in Lemma \ref{lem:Lsmooth}, we have
\begin{equation}\label{eqn:thm_emr_QNN_1}
	\mathcal{L}\left(\bm{\theta}^{(t+1)}\right) -\mathcal{L}\left(\bm{\theta}^{(t)}\right) \leq \left\langle\nabla \mathcal{L}\left(\bm{\theta}^{(t)}\right), \bm{\theta}^{(t+1)}-\bm{\theta}^{(t)} \right\rangle +\frac{S}{2}\left\|\bm{\theta}^{(t+1)}-\bm{\theta}^{(t)}\right\|^2.
\end{equation}

 Combining Eqs.~(\ref{eqn:thm_emr_QNN_0}) and (\ref{eqn:thm_emr_QNN_1}), we obtain
\begin{eqnarray}\label{eqn:thm_emr_QNN_2}
	&& \mathbb{E}_{\bm{x}_i^{(t,w)}\sim \mathcal{D}_i, \bm{\varsigma}^{(t,w)}_{i,j}}\left[ \mathcal{L}\left(\bm{\theta}^{(t+1)}\right) -\mathcal{L}\left(\bm{\theta}^{(t)}\right) \right] \nonumber\\
	 \leq && \mathbb{E}_{\bm{x}_i^{(t,w)}\sim \mathcal{D}_i, \bm{\varsigma}^{(t,w)}_{i,j}}\left[ \left\langle\nabla \mathcal{L}\left(\bm{\theta}^{(t)}\right), \bm{\theta}^{(t+1)}-\bm{\theta}^{(t)} \right\rangle \right] +\frac{S}{2} \mathbb{E}_{\bm{x}_i^{(t,w)}\sim \mathcal{D}_i, \bm{\varsigma}^{(t,w)}_{i,j}}\left[ \left\|\bm{\theta}^{(t+1)}-\bm{\theta}^{(t)}\right\|^2 \right] \nonumber\\
	= &&  \underbrace{ \mathbb{E}_{\bm{x}_i^{(t,w)}\sim \mathcal{D}_i, \bm{\varsigma}^{(t,w)}_{i,j}}\left[ \left \langle\nabla \mathcal{L}\left(\bm{\theta}^{(t)}\right),  - \frac{\eta}{Q} \sum_{i=1}^Q\sum_{w=1}^W g_i^{(t,w)}  \right \rangle \right]}_\text{T2} +\frac{S\eta^2}{2} \underbrace{ \mathbb{E}_{\bm{x}_i^{(t,w)}\sim \mathcal{D}_i, \bm{\varsigma}^{(t,w)}_{i,j}}\left[ \left\|  \frac{1}{Q} \sum_{i=1}^Q\sum_{w=1}^W g_i^{(t,w)} \right\|^2 \right]}_\text{T1}. 
\end{eqnarray}
where the inequality uses the definition of the smooth function and the equality is supported by  Eq.~(\ref{eqn:thm_emr_QNN_1}).

We next derive the upper bounds of the terms T2 and T1 by first removing the data sample noise (i.e., calculating $\mathbb{E}_{\bm{x}_i^{(t,w)}\sim \mathcal{D}_i}[\cdot]$) and then removing the system noise (i.e., calculating $\mathbb{E}_{\bm{\varsigma}_{i,j}^{t,w}}[\cdot]$).  

\textit{\underline{The elimination of the data sample noise.}} For the term T1, we rewrite $g_i^{(t,w)}$ by $\nabla \barL_i(\bm{\theta}_i^{(t,w)})$ to remove the randomness of  sampling training examples. Mathematically, we have
\begingroup
\allowdisplaybreaks
\begin{eqnarray}\label{eqn:proof-thm1-QNN-3}
	&& \mathbb{E}_{\bm{x}_j \sim \mathcal{D}_i}\left[ \left\|  \frac{1}{Q} \sum_{i=1}^Q\sum_{w=1}^W g_i^{(t,w)} \right\|^2 \right] \nonumber\\
	\leq && \frac{W}{Q}  \sum_{i=1}^Q\sum_{w=1}^W \mathbb{E}_{\bm{x}_j \sim \mathcal{D}_i}\left[ \left\|    g_i^{(t,w)} \right\|^2 \right] \nonumber\\ 
	= && \frac{W}{Q}  \sum_{i=1}^Q\sum_{w=1}^W \mathbb{E}_{\bm{x}_j \sim \mathcal{D}_i}\left[ \left\|    g_i^{(t,w)} - \nabla \barL_i\left(\bm{\theta}_i^{(t,w)}\right) + \nabla \barL_i\left(\bm{\theta}_i^{(t,w)}\right) - \nabla \mathcal{L}_i\left(\bm{\theta}_i^{(t,w)}\right) + \nabla \mathcal{L}_i\left(\bm{\theta}_i^{(t,w)}\right) \right\|^2 \right] \nonumber\\ 
	\leq && \frac{2W}{Q}  \sum_{i=1}^Q\sum_{w=1}^W \mathbb{E}_{\bm{x}_j \sim \mathcal{D}_i}\left[ \left\|    g_i^{(t,w)} - \nabla \barL_i\left(\bm{\theta}_i^{(t,w)}\right) + \nabla \barL_i\left(\bm{\theta}_i^{(t,w)}\right) - \nabla \mathcal{L}_i\left(\bm{\theta}_i^{(t,w)}\right)  \right\|^2 \right] + \frac{2W}{Q}  \sum_{i=1}^Q\sum_{w=1}^W \mathbb{E}_{\bm{x}_j \sim \mathcal{D}_i}\left[ \left\|    \nabla \mathcal{L}_i\left(\bm{\theta}_i^{(t,w)}\right)  \right\|^2 \right]  \nonumber \\ 
	= && \frac{4W}{Q} \sum_{i=1}^Q\sum_{w=1}^W  \mathbb{E}_{\bm{x}_j \sim \mathcal{D}_i}\left[ \left\|    g_i^{(t,w)} - \nabla \barL_i\left(\bm{\theta}_i^{(t,w)}\right)  \right\|^2\right] + \frac{4W}{Q} \sum_{i=1}^Q\sum_{w=1}^W  \mathbb{E}_{\bm{x}_j \sim \mathcal{D}_i}\left[ \left\|    \nabla \barL_i\left(\bm{\theta}_i^{(t,w)}\right) - \nabla \mathcal{L}_i\left(\bm{\theta}_i^{(t,w)} \right) \right\|^2\right]  \nonumber\\
	&& +  \frac{2W}{Q}  \sum_{i=1}^Q\sum_{w=1}^W \mathbb{E}_{\bm{x}_j \sim \mathcal{D}_i}\left[ \left\|    \nabla \mathcal{L}_i\left(\bm{\theta}_i^{(t,w)}\right)  \right\|^2 \right] \nonumber\\
	\leq  &&  4W^2 \sigma^2 + \frac{4W}{Q} \sum_{i=1}^Q\sum_{w=1}^W    \underbrace{ \left\|    \nabla \barL_i\left(\bm{\theta}_i^{(t,w)}\right) - \nabla \mathcal{L}_i\left(\bm{\theta}_i^{(t,w)} \right) \right\|^2}_\text{T3}   +  2W^2 G_1^2   ,
\end{eqnarray} 
\endgroup
where the first inequality employs $\|\sum_{i=1}^n \bm{a}_i\|^2 \leq n \sum_{i=1}^n\|\bm{a}_i\|^2$, the second inequality uses the triangle inequality with $\|\bm{a}+\bm{b}\|^2\leq 2\|\bm{a}\|^2+2\|\bm{b}\|^2$, and the last inequality exploits the bounded variance of the estimated gradients in the assumption. 

The derivation of the upper bound of the term T2 is similar to the operations applied to T1, i.e., the estimated gradient $g_i^{(t,w)}$ is substituted with $\nabla \barL_i(\bm{\theta}_i^{(t,w)})$ to remove the randomness of the sampled data. Mathemetically, we have 
\begin{eqnarray}\label{eqn:proof-thm1-QNN-4}
	&& \mathbb{E}_{\bm{x}_j \sim \mathcal{D}_i}\left[ \left \langle\nabla \mathcal{L}\left(\bm{\theta}^{(t)}\right),  - \frac{\eta}{Q} \sum_{i=1}^Q\sum_{w=1}^W g_i^{(t,w)}  \right \rangle \right] \nonumber\\
	 =  && \left \langle\nabla \mathcal{L}\left(\bm{\theta}^{(t)}\right), \mathbb{E}_{\bm{x}_i} \left[ - \frac{\eta}{Q} \sum_{i=1}^Q\sum_{w=1}^W g_i^{(t,w)} \right]  \right  \rangle \nonumber\\
	 = && \left \langle\nabla \mathcal{L}\left(\bm{\theta}^{(t)}\right),   - \frac{\eta}{Q} \sum_{i=1}^Q\sum_{w=1}^W \nabla \barL_i\left(\bm{\theta}_i^{(t,w)}\right)    \right  \rangle \nonumber\\
	 = && - \frac{\eta}{2} \left\|\nabla \mathcal{L}\left(\bm{\theta}^{(t)}\right) \right\|^2  - \frac{\eta}{2} \left\|\frac{1}{Q} \sum_{i=1}^Q\sum_{w=1}^W \nabla  \barL_i\left(\bm{\theta}_i^{(t,w)}\right)  \right\|^2  + \frac{\eta}{2}\left\|\nabla \mathcal{L}\left(\bm{\theta}^{(t)}\right) - \frac{1}{Q} \sum_{i=1}^Q\sum_{w=1}^W \nabla \barL_i\left(\bm{\theta}_i^{(t,w)}\right) \right\|^2 \nonumber\\
	 \leq && - \frac{\eta}{2} \left\|\nabla \mathcal{L}\left(\bm{\theta}^{(t)}\right) \right\|^2     + \frac{\eta}{2}\left\|\nabla \mathcal{L}\left(\bm{\theta}^{(t)}\right) - \frac{1}{Q} \sum_{i=1}^Q\sum_{w=1}^W \nabla \barL_i\left(\bm{\theta}_i^{(t,w)}\right) \right\|^2 \nonumber\\
	 \leq && - \frac{\eta}{2} \left\|\nabla \mathcal{L}\left(\bm{\theta}^{(t)}\right) \right\|^2     + \frac{\eta}{2}\frac{W}{Q}\sum_{i=1}^Q\sum_{w=1}^W \underbrace{ \left\|\nabla \mathcal{L}_i\left(\bm{\theta}_i^{(t,w)}\right) -  \nabla \barL_i\left(\bm{\theta}_i^{(t,w)}\right) \right\|^2}_{T3},
\end{eqnarray} 
where the second equality uses the unbiased estimation property in terms of the sampled data in Eq.~(\ref{append:eqn:eqiv-est-analy}), the third equality uses $\langle \bm{a}, \bm{b}\rangle = \frac{1}{2}(\|\bm{a}\|^2 + \|\bm{b}\|^2- \|\bm{a}-\bm{b}\|^2)$, the second inequality employs the triangle inequality, and the last inequality exploits the explicit form of the analytic gradient such that  
\begin{equation}
	\nabla \mathcal{L}\left(\bm{\theta}^{(t)}\right) = \frac{1}{Q} \sum_{i=1}^Q\sum_{w=1}^W 	\nabla \mathcal{L}_i\left(\bm{\theta}^{(t,w)}\right),
\end{equation}  
and the relation $\|\sum_{i=1}^n \bm{a}_i\|^2 \leq n \sum_{i=1}^n\|\bm{a}_i\|^2$.  

In conjunction with Eqs.~(\ref{eqn:thm_emr_QNN_2}), (\ref{eqn:proof-thm1-QNN-3}), and (\ref{eqn:proof-thm1-QNN-4}), we obtain 
\begin{eqnarray}\label{eqn:thm_emr_QNN_5}
	&& \mathbb{E}_{\bm{x}_i^{(t,w)}\sim \mathcal{D}_i, \bm{\varsigma}^{(t,w)}_{i,j}}\left[ \mathcal{L}\left(\bm{\theta}^{(t+1)}\right) -\mathcal{L}\left(\bm{\theta}^{(t)}\right) \right] \nonumber\\
	 \leq && -\frac{\eta}{2} \left \|\nabla \mathcal{L}\left(\bm{\theta}^{(t)}\right) \right \|^2   +  \frac{\eta W}{2Q} \sum_{i=1}^Q\sum_{w=1}^W \mathbb{E}_{\bm{\varsigma}^{(t,w)}_{i,j}} \underbrace{ \left[  \left \|  \nabla \barL_i\left(\bm{\theta}_i^{(t,w)}\right) - \nabla \mathcal{L}_i\left(\bm{\theta}_i^{(t,w)}\right)  \right \|^2  \right] }_{T3} \nonumber\\
	&& +  \frac{S\eta^2}{2} \left(4W^2 \sigma^2 + 2W^2 G_1^2 \right) + \frac{S\eta^2}{2} \frac{4W}{Q} \sum_{i=1}^Q\sum_{w=1}^W  \underbrace{\mathbb{E}_{\bm{\varsigma}^{(t,w)}_{i,j}}\left[  \left\|    \nabla \barL_i\left(\bm{\theta}_i^{(t,w)}\right) - \nabla \mathcal{L}_i\left(\bm{\theta}_i^{(t,w)} \right) \right\|^2 \right]}_\text{T3}~.    
\end{eqnarray}

\textit{\underline{The elimination of the system and measurement noise.}} By employing Lemma \ref{lem:thm1}, the term T3 can be upper bounded by 
\begin{eqnarray}\label{eqn:thm_emr_QNN_6}
		&& \mathbb{E}_{ \bm{\varsigma}_{i,j}^{(t,w)}}\left[ \left\|    \nabla \barL_i\left(\bm{\theta}_i^{(t,w)}\right) - \nabla \mathcal{L}_i\left(\bm{\theta}_i^{(t,w)} \right) \right\|^2\right] \nonumber\\
		\leq && (\tilde{p}-2)^2\tilde{p}^2 \left\|\nabla \mathcal{L}_i\left(\bm{\theta}_i^{(t,w)}\right) \right\|^2 + \frac{(1-\tilde{p})^2\tilde{p}^2d_Q}{4} + (2-\tilde{p})^2\tilde{p}^2 G_1d_Q + \frac{7(1 - \frac{\tilde{p}}{2})^2 + \frac{1}{8}}{K}d_Q.
	\end{eqnarray}
    
Then, supported by Eqs.~(\ref{eqn:thm_emr_QNN_5}) and (\ref{eqn:thm_emr_QNN_6}), the upper bound in Eq.~(\ref{eqn:thm_emr_QNN_1}) satisfies 
\begin{eqnarray}\label{eqn:thm_emr_QNN_7}
	&& \mathbb{E}_{\bm{x}_j \sim \mathcal{D}_i,\bm{\varsigma}_{i,j}^{(t,w)}}\left[ \mathcal{L}\left(\bm{\theta}^{(t+1)}\right) -\mathcal{L}\left(\bm{\theta}^{(t)}\right) \right] \nonumber\\
	 \leq && -\frac{\eta}{2} \left \|\nabla \mathcal{L}\left(\bm{\theta}^{(t)}\right) \right \|^2  +   \frac{S\eta^2}{2} \left(4W^2 \sigma^2 + 2W^2 G_1^2 \right)    \nonumber\\
	 && + \left( \frac{S\eta^2}{2}\frac{4W}{Q}  + \frac{\eta W}{Q} \right)  \sum_{i=1}^Q\sum_{w=1}^W  \mathbb{E}_{\bm{x}_j \sim \mathcal{D}_i,\bm{\varsigma}_{i,j}^{(t,w)}}\left[ \underbrace{ \left\|    \nabla \barL_i\left(\bm{\theta}_i^{(t,w)}\right) - \nabla \mathcal{L}_i\left(\bm{\theta}_i^{(t,w)} \right) \right\|^2}_\text{T3} \right] \nonumber\\
	\leq &&   -\frac{\eta}{2} \left \|\nabla \mathcal{L}\left(\bm{\theta}^{(t)}\right) \right \|^2   +   \frac{S\eta^2}{2} \left(4W^2 \sigma^2 + 2W^2 G_1^2 \right)     \nonumber\\
	&&+ \left( \frac{S\eta^24W^2}{2} + \eta W^2 \right)  \left((\tilde{p}-2)^2\tilde{p}^2 G_1^2 +   \frac{(1-\tilde{p})^2\tilde{p}^2d_Q}{4} + (2-\tilde{p})^2\tilde{p}^2 G_1d_Q + \frac{7(1 - \frac{\tilde{p}}{2})^2 + \frac{1}{8}}{K}d_Q \right)~.
\end{eqnarray}

Rearranging the terms in Eq.~(\ref{eqn:thm_emr_QNN_7}), the norm of the gradients $\|\nabla \mathcal{L} (\bm{\theta}^{(t)}\|$ is upper bounded by 
\begin{eqnarray}\label{eqn:thm_emr_QNN_8}
	 \left \|\nabla \mathcal{L}\left(\bm{\theta}^{(t)}\right) \right \|^2 \leq && \frac{2}{\eta} \mathbb{E}_{\bm{x}_j \sim \mathcal{D}_i,\bm{\varsigma_i}}\left[   \mathcal{L}\left(\bm{\theta}^{(t)}\right) - \mathcal{L}\left(\bm{\theta}^{(t+1)}\right) \right] + S\eta \left(4W^2 \sigma^2 + 2W^2 G_1^2 \right)  \nonumber \\
	 +  &&    \left( S\eta4W^2 + 2 W^2 \right)  \left((\tilde{p}-2)^2\tilde{p}^2 G_1^2 +   \frac{(1-\tilde{p})^2\tilde{p}^2d_Q}{4} + (2-\tilde{p})^2\tilde{p}^2 G_1d_Q + \frac{7(1 - \frac{\tilde{p}}{2})^2 + \frac{1}{8}}{K}d_Q \right).
\end{eqnarray}

Summing over $t$ and dividing both sides by $T$ in Eq.~(\ref{eqn:thm_emr_QNN_8}), we obtain
\begin{eqnarray}
&&  \frac{1}{T}\sum_{t=1}^T  \left \|\nabla \mathcal{L}\left(\bm{\theta}^{(t)}\right) \right \|^2 \nonumber\\
\leq &&  \frac{2}{\eta} \mathbb{E}_{\bm{x}_j \sim \mathcal{D}_i,\bm{\varsigma_i}}\left[   \mathcal{L}\left(\bm{\theta}^{(1)}\right) - \mathcal{L}\left(\bm{\theta}^{(T+1)}\right) \right] + S\eta \left(4W^2 \sigma^2 + 2W^2 G_1^2 \right) \nonumber \\
	 +  &&    \left( S\eta4W^2 + 2 W^2 \right)  \left((\tilde{p}-2)^2\tilde{p}^2 G_1^2 +   \frac{(1-\tilde{p})^2\tilde{p}^2d_Q}{4} + (2-\tilde{p})^2\tilde{p}^2 G_1d_Q + \frac{7(1 - \frac{\tilde{p}}{2})^2 + \frac{1}{8}}{K}d_Q \right)  \nonumber\\
\leq && \frac{2+80\lambda d_Q}{\eta T}   +       S\eta  \left(4W^2 \sigma^2 + 2W^2 G_2^2 \right) \nonumber\\
    && + \left(  4S\eta W^2  + 2W^2 \right)  \left((\tilde{p}-2)^2\tilde{p}^2 G_1^2 +   \frac{(1-\tilde{p})^2\tilde{p}^2d_Q}{4} + (2-\tilde{p})^2\tilde{p}^2 G_1d_Q + \frac{7(1 - \frac{\tilde{p}}{2})^2 + \frac{1}{8}}{K}d_Q \right),
\end{eqnarray}
where the second inequality exploits the upper bound of the discrepancy of loss function. 

With setting $\eta=\sqrt{1/ST}$, we achieve
\begin{eqnarray}
&&  \frac{1}{T}\sum_{t=1}^T  \left \|\nabla \mathcal{L}\left(\bm{\theta}^{(t)}\right) \right \|^2  \nonumber\\
\leq && (2+80\lambda d_Q) \sqrt{\frac{S}{T}}  +       \sqrt{\frac{S}{T}}  \left(4W^2 \sigma^2 + 2W^2 G_2^2 \right) \nonumber\\
    && + \left(  4  W^2 \sqrt{\frac{S}{T}}   + 2W^2 \right)  \left((\tilde{p}-2)^2\tilde{p}^2 G_1^2 +   \frac{(1-\tilde{p})^2\tilde{p}^2d_Q}{4} + (2-\tilde{p})^2\tilde{p}^2 G_1d_Q + \frac{7(1 - \frac{\tilde{p}}{2})^2 + \frac{1}{8}}{K}d_Q \right).
\end{eqnarray} 

\end{proof}
 
\subsection{Proof of Lemma \ref{lem:thm1}}\label{append:subsec:proof-lem2}
 As shown in \cite{du2020learnability}, the estimated gradients, which are caused by the gates noise and the sample errors, can be explicitly formulated to relate with its analytic gradients.  
\begin{lemma}[Modified from Theorem 3, \cite{du2020learnability}]\label{lem:noise_QNN_gaussian}
Denote $\tilde{p}=1-(1-p)^{L_Q}$ with $L_Q$ being the quantum circuit depth. At the $(t,w)$-th iteration, we  define five constants with
 \[
    C^{(t,w)}_{i,j,a}= 
\begin{cases}
    (1-\tilde{p})\tilde{p}(\frac{1}{2}-{y}_i^{(t,w)})(\hat{y}_{i}^{(t,w,+_j)}-\hat{y}_{i}^{(t,w,-_j)}) - (2\tilde{p}-\tilde{p}^2)\lambda\bm{\theta}_j^{(t, w)},&  a=1\\
    (1-\tilde{p})\left(\hat{y}_{i}^{(t,w,+_j)}-\hat{y}_{i}^{(t,w,-_j)} \right),  & a=2 \\
 (1-\tilde{p})\hat{y}_i^{(t,w)} + \frac{\tilde{p}}{2}- y_i^{(t,w)},   & a=3  \\
 \frac{-(1-\tilde{p})(\hat{y}_i^{(t,w)})^2 + (1-\tilde{p})^2\hat{y}_i^{(t,w)}+\frac{\tilde{p}}{2}-\frac{\tilde{p}^2}{4}}{K}, & a=4 \\
   \frac{-(1-\tilde{p})((\hat{y}_{i,+_j}^{(t,w)})^2+(\hat{y}_{i,-_j}^{(t,w)})^2) + (1-\tilde{p})^2(\hat{y}_{i,+_j}^{(t, w)}+\hat{y}_{i,-_j}^{(t, w)})+\tilde{p}-\frac{\tilde{p}^2}{2}}{K}, & a =5,
\end{cases}
\]
where $\hat{y}_{i}^{(t,w,\pm_j)}$, $\hat{y}_i^{(t,w)}$, and $y_i^{(t,w)}$ are defined in Eq.~(\ref{eqn:est-grad-QNN}), $K$ refers to the number of quantum  measurements.  

 The relation between the estimated and analytic gradients of QNN follows  
 \begin{equation}\label{eqn:analy-est-grad-QNN}
 \nabla_j\barL_i(\bm{\theta}^{(t,w)}) = (1-\tilde{p})^2\nabla_j{\mathcal{L}}_i(\bm{\theta}^{(t,w)}) + C^{(t,w)}_{i,j,1} +  \bm{\varsigma}_{i,j}^{(t,w)}, \forall j\in[d]	
 \end{equation}
 with $\bm{\varsigma}_{i,j}^{(t,w)}= C^{(t,w)}_{i,j,2}\xi_i^{(t,w)} + C^{(t,w)}_{i,j,3}\xi_{i,j}^{(t,w)}+\xi^{(t,w)}\xi_{i,j}^{(t,w)}$, where $\xi_{i}^{(t,w)}$ and $\xi_{i,j}^{(t,w)}$ are two random variables with zero mean and variances $C^{(t,w)}_{i,j,4}$ and $C^{(t,w)}_{i,j,5}$, respectively. 
\end{lemma}
This results can be used to obtain Lemma \ref{lem:thm1}.

\begin{proof}[Proof of Lemma \ref{lem:thm1}]
	We first derive the term  	\begin{eqnarray}
		&& \mathbb{E}_{\bm{\varsigma}_{i,j}^{(t,w)}}\left[ \left\|    \nabla \barL_i\left(\bm{\theta}_i^{(t,w)}\right) - \nabla \mathcal{L}_i\left(\bm{\theta}_i^{(t,w)} \right) \right\|^2\right] \nonumber\\
		= && \sum_{j=1}^{d_Q} \mathbb{E}_{\bm{\varsigma}_{i,j}^{(t,w)}}\left[   \left( \nabla_j \barL_i\left(\bm{\theta}_i^{(t,w)}\right)  - \nabla_j \mathcal{L}_i\left(\bm{\theta}_i^{(t,w)}\right) \right)^2 \right] \nonumber\\
		= && \sum_{j=1}^{d_Q} \mathbb{E}_{\bm{\varsigma}_{i,j}^{(t,w)}}\left[   \left( (\tilde{p}-2)\tilde{p}\nabla_j \mathcal{L}_i\left(\bm{\theta}_i^{(t,w)}\right) + C^{(t,w)}_{i,j,1} +  \bm{\varsigma}_{i,j}^{(t,w)} \right)^2  \right] \nonumber\\
		= && \sum_{j=1}^{d_Q} \left( (\tilde{p}-2)\tilde{p}\nabla_j \mathcal{L}_i\left(\bm{\theta}_i^{(t,w)}\right)\right)^2 + \left(C^{(t,w)}_{i,j,1}\right)^2 +  \mathbb{E}_{\bm{x}_i \sim \bm{\varsigma}_{i,j}^{(t,w)}}\left[ \left(\bm{\varsigma}_{i,j}^{(t,w)} \right)^2 \right] + 2(\tilde{p}-2)\tilde{p}\nabla_j \mathcal{L}_i\left(\bm{\theta}_i^{(t,w)}\right)C^{(t,w)}_{i,j,1} \nonumber\\
		 && ~~~~+  2\left((\tilde{p}-2)\tilde{p}\nabla_j \mathcal{L}_i\left(\bm{\theta}_i^{(t,w)}\right)+C^{(t,w)}_{i,j,1}\right)\mathbb{E}_{\bm{x}_i \sim \bm{\varsigma}_{i,j}^{(t,w)}}\left[ \bm{\varsigma}_{i,j}^{(t,w)}   \right] \nonumber\\
		 \leq && (\tilde{p}-2)^2\tilde{p}^2 \left\|\nabla \mathcal{L}_i\left(\bm{\theta}_i^{(t,w)}\right) \right\|^2 + \frac{(1-\tilde{p})^2\tilde{p}^2d_Q}{4} +  \mathbb{E}_{\bm{x}_i \sim \bm{\varsigma}_{i,j}^{(t,w)}}\left[ \left(\bm{\varsigma}_{i,j}^{(t,w)} \right)^2 \right]   + (2-\tilde{p})(1-\tilde{p})\tilde{p}^2 G_1d_Q\nonumber\\
		 = && (\tilde{p}-2)^2\tilde{p}^2 \left\|\nabla \mathcal{L}_i\left(\bm{\theta}_i^{(t,w)}\right) \right\|^2 + \frac{(1-\tilde{p})^2\tilde{p}^2d_Q}{4} + \sum_{j=1}^{d_Q}  \left(C^{(t,w)}_{i,j,2}\right)^2C^{(t,w)}_{i,j,4} + \left(C^{(t,w)}_{i,j,3}\right)^2C^{(t,w)}_{i,j,5}+ C^{(t,w)}_{i,j,4}C^{(t,w)}_{i,j,5} \nonumber\\
		 && +  (2-\tilde{p})(1-\tilde{p})\tilde{p}^2 G_1d_Q \nonumber\\
		 \leq && (\tilde{p}-2)^2\tilde{p}^2 \left\|\nabla \mathcal{L}_i\left(\bm{\theta}_i^{(t,w)}\right) \right\|^2 + \frac{(1-\tilde{p})^2\tilde{p}^2d_Q}{4} +  (1-\tilde{p})^2 d_Q \frac{(1-\tilde{p})^2 + 1/4}{K} + \left(1 - \frac{\tilde{p}}{2}\right)^2d_Q\frac{2(1-\tilde{p})^2+1/2}{K} \nonumber\\
		 && + \frac{(2(1-\tilde{p})^2+1/2)((1-\tilde{p})^2 + 1/4)d_Q}{K^2} + (2-\tilde{p})(1-\tilde{p})\tilde{p}^2 G_1d_Q \nonumber\\
		 \leq && (\tilde{p}-2)^2\tilde{p}^2 \left\|\nabla \mathcal{L}_i\left(\bm{\theta}_i^{(t,w)}\right) \right\|^2 + \frac{(1-\tilde{p})^2\tilde{p}^2d_Q}{4} + (2-\tilde{p})^2\tilde{p}^2 G_1d_Q + \frac{7(1 - \frac{\tilde{p}}{2})^2 + \frac{1}{8}}{K}d_Q ~,
	\end{eqnarray}
	where the first equality comes from the definition of the $l_2$ norm, the second equality uses Eq.~(\ref{eqn:analy-est-grad-QNN}), the first inequalitu employs  $\mathbb{E}_{\bm{\varsigma}_{i,j}^{(t,w)}}[ \bm{\varsigma}_{i,j}^{(t,w)}]=0$ and the Lipschitz property of $\mathcal{L}_i$ in Lemma \ref{lem:Lsmooth}, the second inequality utilizes the upper bounds of $C^{(t,w)}_{i,j,a}$ in Lemma \ref{lem:noise_QNN_gaussian}, and the last inequality simplifies the factor $\tilde{p}$.   
	
\end{proof}

\section{Numerical simulation details of QUDIO towards hand-written digits image classification}\label{app:qnn}
This section presents the detailed description about applying QUDIO to accomplish hand-written digits image classification tasks. First, in  Appendix \ref{append:subsec:QNN-data}, the construction of local nodes is explained. Then, the hyper-parameter configuration of QUDIO is introduced in Appendix \ref{append:subsec:QNN-hyper}. Last, in Appendix \ref{append:subsec:sim-res-QNN}, we provide complementary simulation results of QUDIO together with thorough discussions.

\subsection{The setup of local nodes}\label{append:subsec:QNN-data}
\textbf{The distilled dataset.} We  sample $756$ digit images labeled with `0' and `1' from the MNIST handwritten digit database \cite{lecun1998mnist}, where $256$ images (examples) compose the  training set and the rest $500$ images (examples) form the test dataset. Once the training and test sets are collected, the data preprocessing is applied, i.e., all examples are down-sampled to $8\times 8$ pixels followed by the vectorization and $l_2$ normalization. 

Given the processed training set $\mathcal{D}=\{\bm{x}_i, y_i\}_{i=1}^{n}$ with $n=256$ and $\bm{x}_i\in\mathbb{R}^{64}$, the central server partitions it into $\{\mathcal{D}_i\}_{i=1}^Q$ and allocate them to $Q$ local nodes. For example, when $Q=32$, we have $|\mathcal{D}_i|=8$ for $\forall i\in[32]$.  

\textbf{The implementation of local nodes.} Due to the identical implementation of all local nodes, here we mainly consider the setup of the $i$-th local node $\mathcal{Q}_i$. The construction of PQCs of $\mathcal{Q}_i$, i.e., $h(\bm{\theta},O,\rho_i) = \Tr(OU(\bm{\theta})\rho_iU(\bm{\theta})^{\dagger}),~\forall i\in[n]$ in Eq.~(\ref{eqn:QNN-ansatz}),  is illustrated in Fig.~\ref{fig:ansatz}. Particularly, the amplitude encoding method is adopted to encode $\bm{x}_i\in\mathcal{D}_i$ to the quantum state $\rho_i$, i.e.,   
\begin{equation}
 \rho_i=   \ket{\bm{x}_i}\bra{\bm{x}_i}, ~\text{and}~  \ket{\bm{x}_i} = \sum_{j=1}^{64} \bm{x}_{i,j}\ket{j} . 
\end{equation}
The prepared state $ \rho_i$ is interacted with $U(\bm{\theta})$ implemented by the hardware-efficient ansatz. Note that we set the block number in  $U(\bm{\theta})$ as $L=4$ and each block contains a single-qubit layer, i.e., $\otimes_{i=1}^{N} Rot(\bm{\theta}_i)=\otimes_{i=1}^{N} \RZ (\bm{\theta}_{i,1})\RY (\bm{\theta}_{i,2})\RZ (\bm{\theta}_{i,3})$ and an entangled layer formed by CNOT gates, as highlighted in the dashed box of Fig.~\ref{fig:ansatz}. The quantum measurement operator $O$ is set as $O=\mathbb{I}_{2^5}\otimes \ket{0}\bra{0}$. 

The predicted label of $\bm{x}_i$ is assigned as `$0$'  if $h(\bm{\theta},O,\rho_i)\leq 0.5$; otherwise, $\bm{x}_i$ is classified as `$1$'. In the NISQ setting,  $\bm{x}_i$ is predicted as `$0$' if the sample mean follows $\bar{y}_i\leq 0.5$; otherwise, $\bm{x}_i$ is classified as `$1$'.   

\begin{figure}[htp]
\centering
\includegraphics[width=0.75\textwidth]{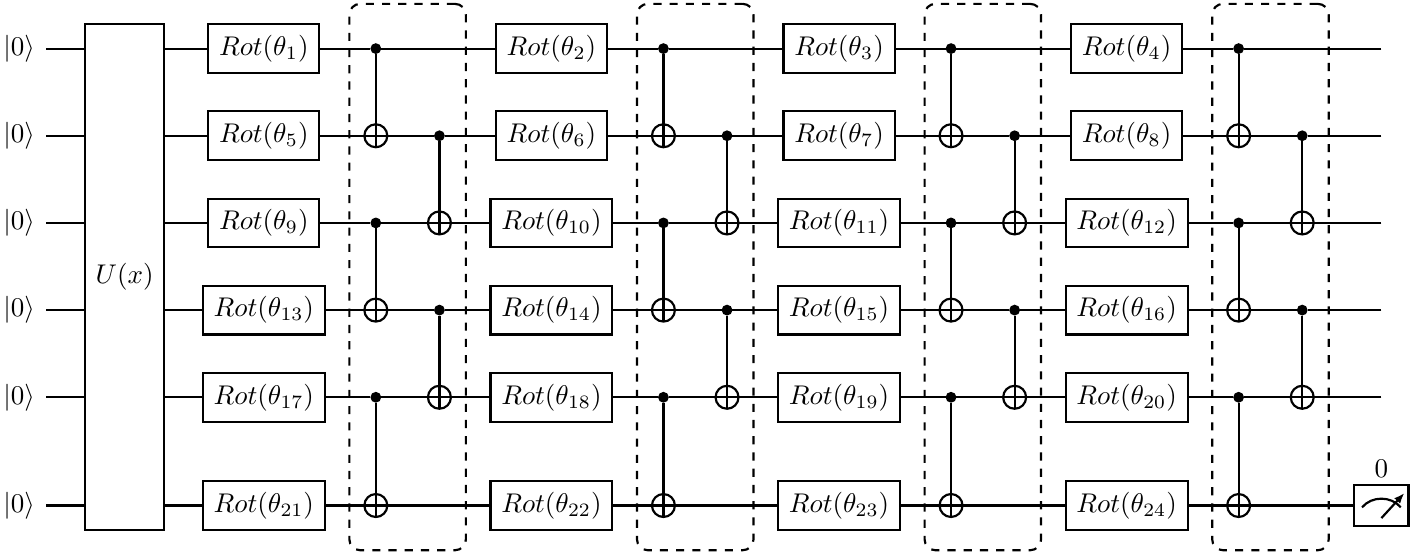}	
\caption{\small{\textbf{The implementation of the local node $\mathcal{Q}_i$ in the binary classification task.} The  amplitude embedding method is exploited to convert the classical example $\bm{x}_i$ into the quantum state $\ket{\bm{x}_i}$. The hardware-efficient ansatz is adopted to implement the trainable unitary $U(\bm{\theta})$. The quantum observable is set as $O=\mathbb{I}_{2^5}\otimes \ket{0}\bra{0}$.}}
\label{fig:ansatz}
\end{figure}

\subsection{Hyper-parameter settings}\label{append:subsec:QNN-hyper}

\textbf{The source code of QUDIO}. The realization of QUDIO is based on Pytorch \cite{NEURIPS2019_9015} and its distributed communication package. To be more concrete, we select the GLOO backend and ring all-reduce operations to achieve the communication protocol between processes on CPU. Note that this distributed optimization framework can be easily extended to coordinate  multiple quantum processors.

\textbf{The classical optimizers.} In QUDIO, the stochastic gradient descent optimizer is adopted to update parameters. Its hyper-parameters settings are as follows. The initial learning rate is set as $0.01$, the momentum factor is set as $0.9$, and the decay-rate is set as $0.1$ every $40$ epochs. For the ideal case, the gradient of each learnable parameter is calculated by back propagation \cite{rumelhart1986learning}. For the NISQ case, the depolarizing noise and measurement error are introduced and the gradient is estimated by parameter shift rule. 

\textbf{Hardware parameters.}  All  simulation results in this study are completed by the classical device with Intel(R) Xeon(R) Gold 6267C CPU @ 2.60GHz and 128 GB memory.
 
\begin{figure*} 
\captionsetup[subfigure]{justification=centering}
\centering
\subfloat[]{
\centering
\includegraphics[width=0.38\textwidth]{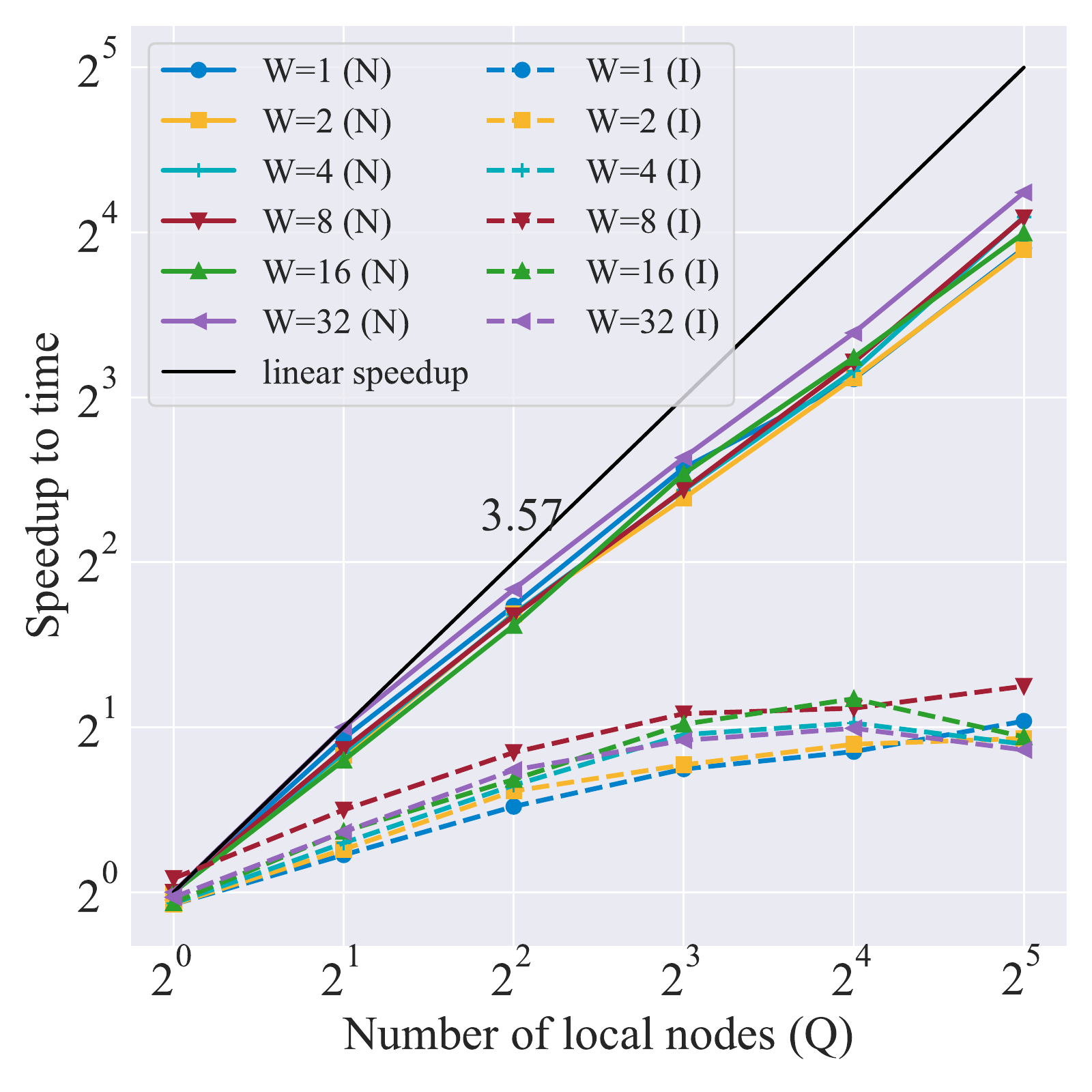}}
\subfloat[]{
\centering
\includegraphics[width=0.38\textwidth]{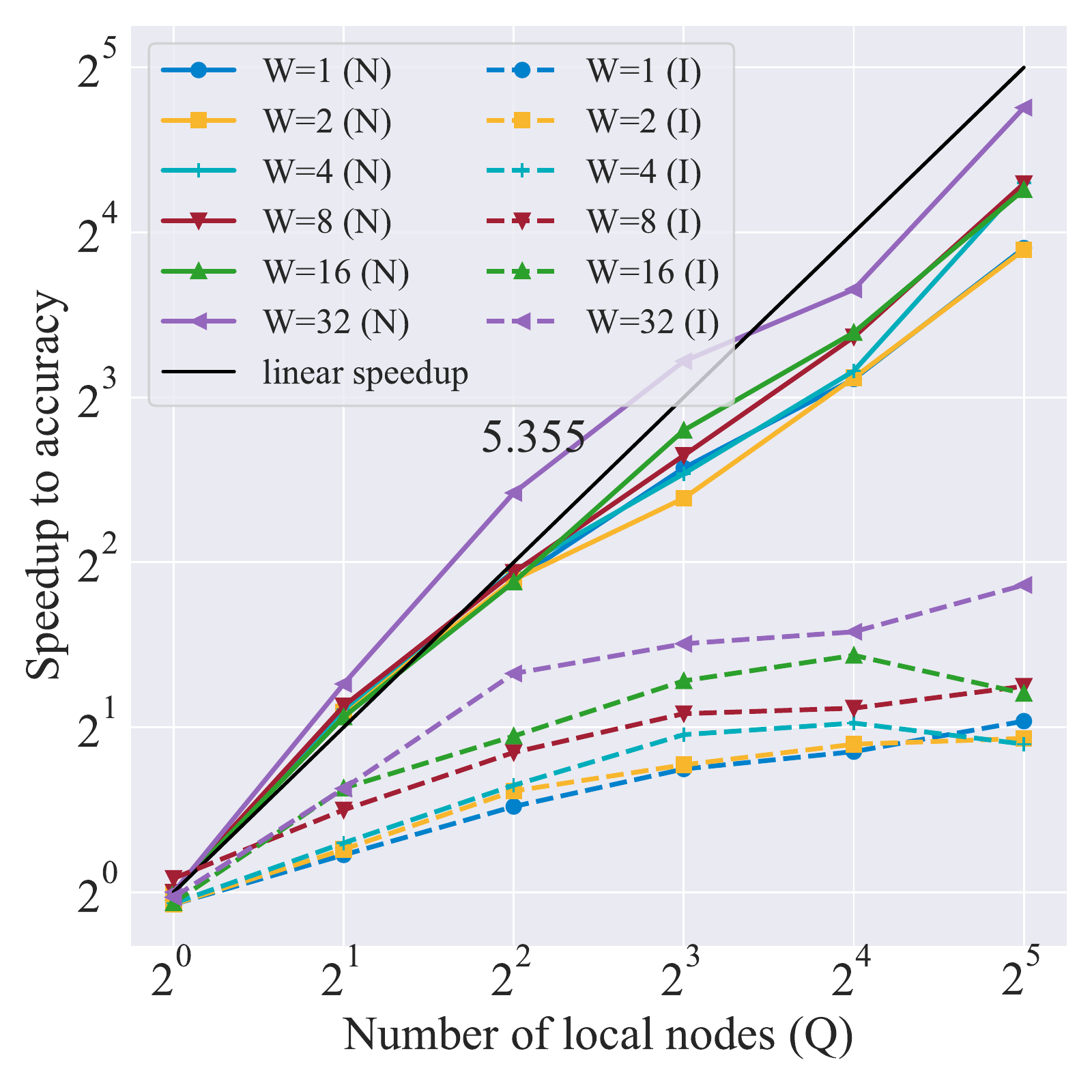}}
\caption{\small{\textbf{QUDIO accelerates the training of QNNs.} Each point refers to the averaged result with five repetitions. The labels `N' and `I' represent the noisy and ideal scenarios, respectively. (a) The speedup ratio of QUDIO in the measure of a fixed number of training iterations $T$. (b) The speedup ratio of QUDIO in the measure of the time to accuracy.}}
\label{fig:app:qnn_speedup}
\end{figure*}

\subsection{More simulation results}\label{append:subsec:sim-res-QNN}
Here we demonstrate more simulation results about QUDIO that are omitted in the main text. Specifically, we first evaluate the speedup ratio of QUDIO in a fine-grained setting. Then, we investigate how the number of local nodes $Q$ effects the learning performance. Finally, we study how the system noise and the number of measurements influence the performance of  QUDIO. These supplementary results facilitate us to better understand the capability of QUDIO from different views.      

\textbf{Speedup analysis.} The implementation detail of QUDIO in the exploration of the speedup ratio is as follows. Both the number of local nodes and the number of local iterations have six varied settings, i.e., $Q, W\in\{1, 2, 4, 8, 16, 32\}$. In the NISQ scenario, the depolarization rate and the number of measurement are set as $p=0.0001$ and $K=100$, respectively.

Fig.~\ref{fig:app:qnn_speedup}(a) and Fig.~\ref{fig:app:qnn_speedup}(b) respectively depict the speedup ratio of QUDIO in the measure of a fixed number of training iterations $T$ and the time to accuracy, and Tab.~\ref{tab:qnn:run_time_q} and Tab.~\ref{tab:qnn:run_time_w} record the concrete running time under each setting. Note that the first metric differs with the second one in the sense that the latter requires QUDIO to surpass a threshold accuracy instead of a fixed $T$. For both two metrics, QUDIO attains a sublinear speedup in terms of $Q$. Besides, a larger number of local iterations $W$ generally promises a higher speedup ratio. We also notice that there exists a manifest margin between the noiseless and NISQ settings. This phenomenon is mainly caused by the opposite role of the communication overhead, i.e., the communication overhead occupies a large portion of the  computational runtime in the noiseless case, while it becomes negligible in the NISQ case.

\begin{table}[htp]
    \captionsetup[subtable]{justification=centering}
    \centering
    \subfloat[Time to iteration]{
    \begin{tabular}{ccc}
    \hline
    Number of local nodes ($Q$) & NISQ & Running time (s) \\
    \hline
        1 & \XSolidBrush & 1742.21 \\
        1 & \checkmark & 41974.36 \\
        2 & \XSolidBrush & 1324.12 \\
        2 & \checkmark & 20992.67 \\
        4 & \XSolidBrush & 1020.48 \\
        4 & \checkmark & 11756.49 \\
        8 & \XSolidBrush & 900.62 \\
        8 & \checkmark & 6761.86 \\
        16 & \XSolidBrush & 856.47 \\
        16 & \checkmark & 4002.11 \\
        32 & \XSolidBrush & 938.07 \\
        32 & \checkmark & 2218.29 \\
    \hline
    \end{tabular}}
    \quad
    \subfloat[Time to accuracy]{
    \begin{tabular}{ccc}
    \hline 
    Number of local nodes ($Q$) & NISQ & Running time (s) \\
    \hline
        1 & \XSolidBrush & 104.53 \\
        1 & \checkmark & 2518.46 \\
        2 & \XSolidBrush & 66.20 \\
        2 & \checkmark & 1049.63 \\
        4 & \XSolidBrush & 40.81 \\
        4 & \checkmark & 470.25 \\
        8 & \XSolidBrush & 36.02 \\
        8 & \checkmark & 270.47 \\
        16 & \XSolidBrush & 34.25 \\
        16 & \checkmark & 200.10 \\
        32 & \XSolidBrush & 28.14 \\
        32 & \checkmark & 93.16 \\
    \hline
    \end{tabular}}
    \caption{\small{\textbf{The runtime of QUDIO with respect to the varied number of local nodes $Q$ and fixed $W=32$ local steps}}.}
    \label{tab:qnn:run_time_q}
\end{table}
  
\begin{table}[htp]
    \captionsetup[subtable]{justification=centering}
    \centering
    \subfloat[Time to iteration]{
    \begin{tabular}{ccc}
    \hline
    Number of local steps ($W$) & NISQ & Running time (s) \\
    \hline
        1 & \XSolidBrush & 906.21 \\
        1 & \checkmark & 2833.21 \\
        2 & \XSolidBrush & 930.84 \\
        2 & \checkmark & 2589.63 \\
        4 & \XSolidBrush & 870.69 \\
        4 & \checkmark & 2253.47 \\
        8 & \XSolidBrush & 761.29 \\
        8 & \checkmark & 2215.43 \\
        16 & \XSolidBrush & 872.54 \\
        16 & \checkmark & 2354.78 \\
        32 & \XSolidBrush & 938.07 \\
        32 & \checkmark & 2218.29 \\
    \hline
    \end{tabular}}
    \quad
    \subfloat[Time to accuracy]{
    \begin{tabular}{ccc}
    \hline
    Number of local steps ($W$) & NISQ & Running time (s) \\
    \hline
        1 & \XSolidBrush & 54.37 \\
        1 & \checkmark & 169.99 \\
        2 & \XSolidBrush & 55.85 \\
        2 & \checkmark & 155.37 \\
        4 & \XSolidBrush & 52.24 \\
        4 & \checkmark & 117.18 \\
        8 & \XSolidBrush & 45.67 \\
        8 & \checkmark & 115.20 \\
        16 & \XSolidBrush & 43.62 \\
        16 & \checkmark & 117.73 \\
        32 & \XSolidBrush & 28.14 \\
        32 & \checkmark & 93.16 \\
    \hline
    \end{tabular}}
    \caption{\small{\textbf{The runtime of QUDIO with respect to the varied number of local steps $W$ and fixed $Q=32$ local nodes.}}}
    \label{tab:qnn:run_time_w}
\end{table}

\textbf{Accuracy analysis.} We next turn to explore how the factors $Q$ and $W$ influence the final test accuracy. Analogous to the speedup analysis, the number of local nodes and local iterations have six settings, i.e., $Q, W\in\{1, 2, 4, 8, 16, 32\}$. Meanwhile, the depolarizing rate and the number of measurement are set as  $p=0.0001$ and $K=100$ respectively. 

As shown in Fig.~\ref{fig:app:qnn_acc}, for all settings, the test accuracy achieved by QUDIO is above $95\%$. This observation reflects the robustness of QUDIO. Moreover, when the number of local iterations $W$ is kept to be identical, QUDIO gains the highest test accuracy with the setting $Q=16$. On the contrary, when the number of local nodes $Q$ is kept to be identical, increasing $W$ subsumes to a deteriorate test accuracy, which complies with the claim of Theorem \ref{thm:conv-qnn}.
 
\begin{figure*}[htp]
\centering
\includegraphics[width=0.82\textwidth]{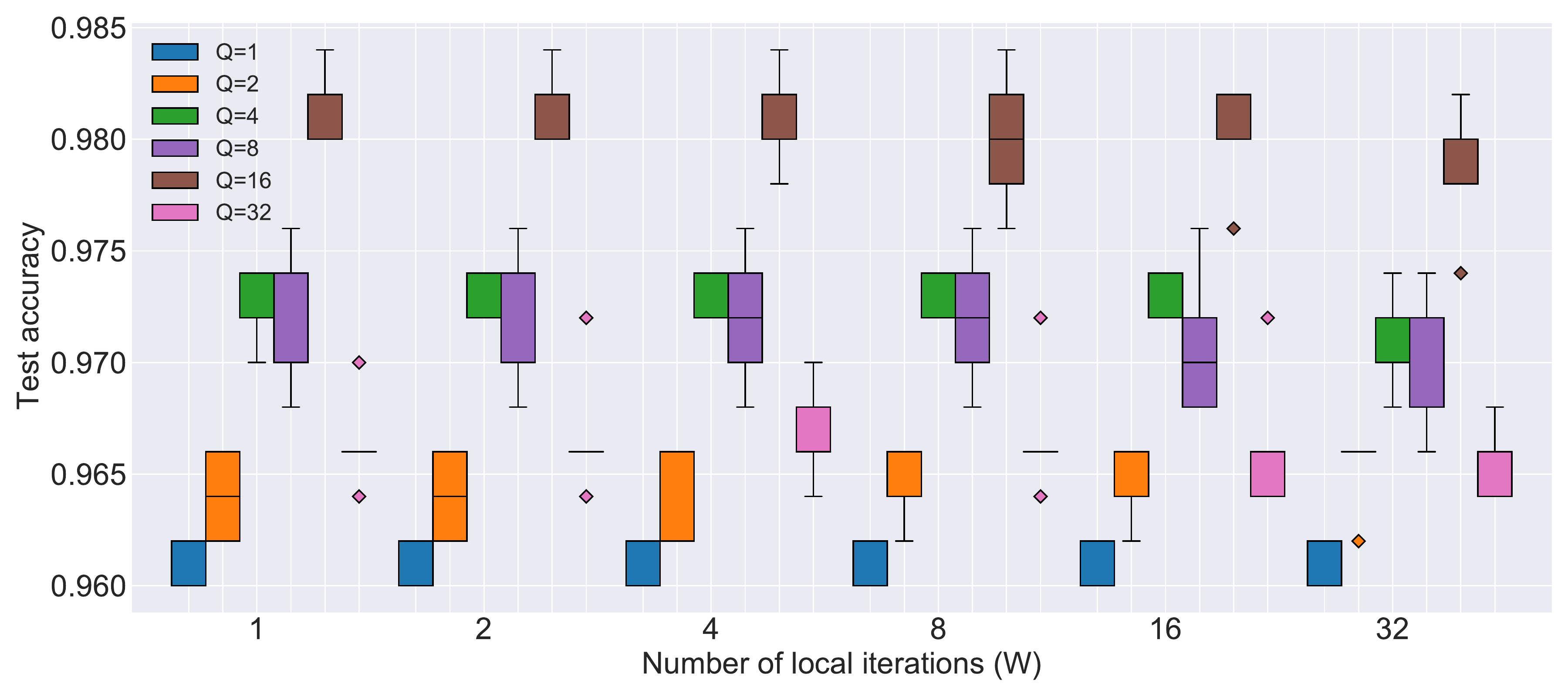}
\caption{\small{\textbf{Role of number of quantum processors and local iterations.} The box plot illustrate the final test accuracy corresponding to varied settings of $Q$ and $W$. }}
\label{fig:app:qnn_acc}
\end{figure*}

\textbf{The role of the system noise and the number of measurements.} We end this section by comprehending how the factors $p$ and $K$ influence the performance of QUDIO. In particular, the depolarization noise rate $p$ scales from $0.0001$ to $0.0512$, and the number of measurements  ranges from $5$ to $100$. The number of local nodes and local iterations is fixed to be $Q=16$ and $W=2$, respectively. 

Fig.~\ref{fig:app:qnn_noise} summarizes the simulation results. In particular, when  $p<0.0064$, the performance of QUDIO heavily depends on the number of measurements. For example, the test accuracy is around $98\%$ with $M=100$, while it drops to $87\%$  with $M=5$. When $p>0.0064$, both $p$ and $K$ determine the performance of QUDIO. For example, for the setting $K=5$, the test accuracy of QUDIO with $p=0.0512$ is reduced by $14\%$ than the setting $p=0.0256$ (i.e., from $77\%$ to $66\%$).

\begin{figure*}[htp]
\centering
\includegraphics[width=1.0\textwidth]{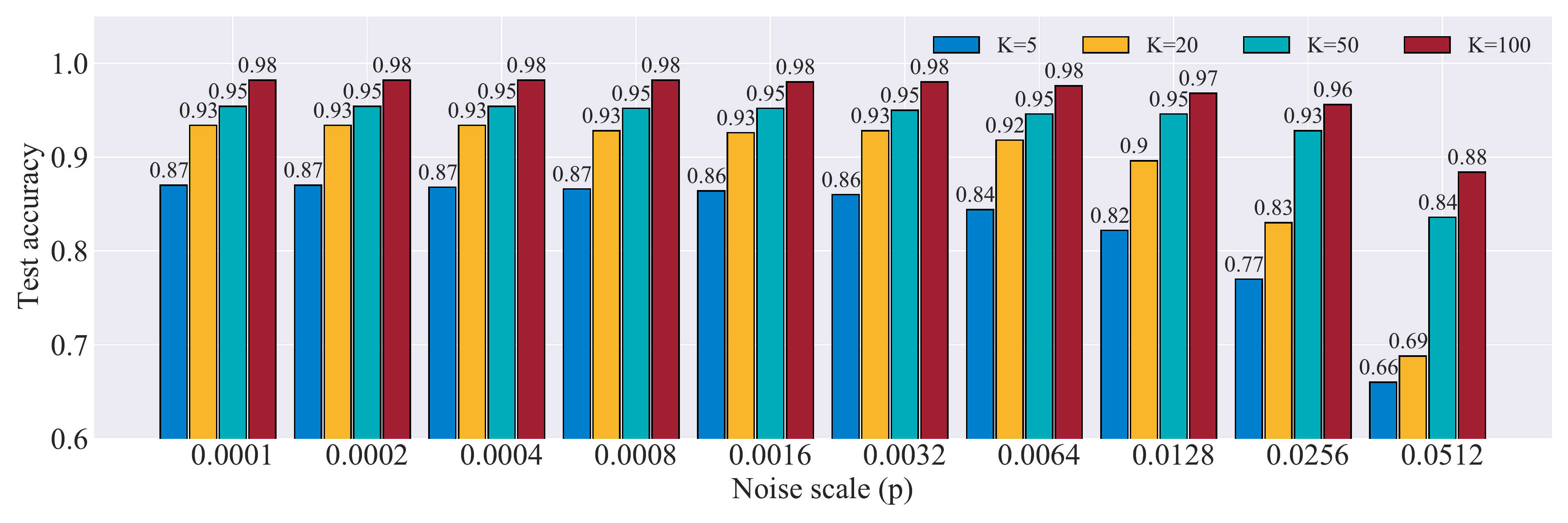}
\caption{\small{\textbf{The test accuracy of QUDIO with varied depolarization rate and the number of measurements.} The label `$K=a$' refers that the number of measurements $K$ is set as $a$. }}
\label{fig:app:qnn_noise}
\end{figure*}

\section{Numerical simulation details of QUDIO towards the ground state energy estimation of hydrogen molecule}\label{app:vqe}
This section provides elaboration about applying QUDIO to estimate the ground state energy of hydrogen molecule tasks. First, the setup of local nodes and the hyper-parameters settings are shown in Appendix \ref{append:VQE-local-setup}. Then, we provide complementary simulation results of QUDIO together with thorough discussions in Appendix \ref{append:VQE-sim}.

\subsection{Implementation of local nodes and hyper-parameters setting}\label{append:VQE-local-setup}
The implementation of QUDIO mainly follows the proposal \cite{kandala2017hardware}. Namely, the binary tree encoding method \cite{bravyi2002fermionic} is used to map the hydrogen molecular Hamiltonian into a 4-qubit system, where $H\in\mathbb{C}^{2^4\times 2^4}$ consists of $n=15$ local Hamiltonian terms. In QUDIO, the central server partitions these $15$ local terms into $Q$ subgroups $\{\mathcal{S}_i\}$ and allocate them to $Q$ local nodes.    

The realization of all local nodes follows the same routine. With this regard, here we only discuss the realization of the node $\mathcal{Q}_i$. The input quantum state is modified to $\rho_0=\ket{1100}\bra{1100}$. The implementation of the ansatz $U(\bm{\theta})$ is shown in Fig.~\ref{fig:app:vqe_ansatz}(a), which is formed by $4$ trainable single-qubit gates followed by $3$ CNOT gates. The prepared state is continuously operated with the observable $H_{S_i}$ to proceed optimization.

\begin{figure*}[htp]
\captionsetup[subfigure]{justification=centering}
\centering
\includegraphics[width=1.0\textwidth]{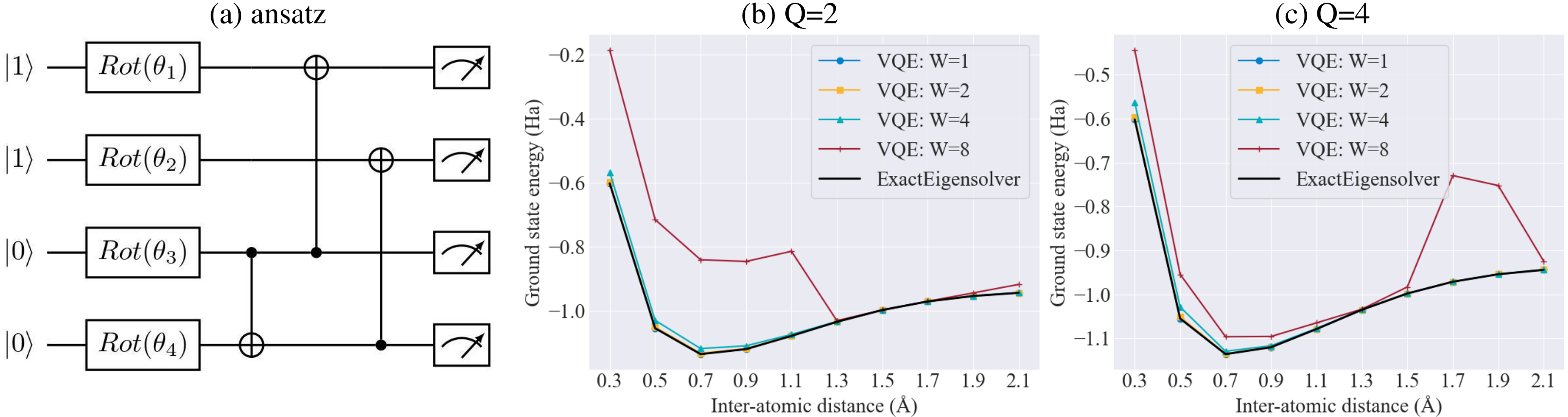}
\caption{\small{\textbf{The implementation of local nodes in QUDIO and simulation results. (a) The circuit implementation in the local node. (b)-(c) The estimated potential energy surface achieved by QUDIO when $Q=2$ and $Q=4$.} }}
\label{fig:app:vqe_ansatz}
\end{figure*}

\subsection{More simulation results}\label{append:VQE-sim}
We conduct extensive numerical simulations to benchmark how the number of local nodes $Q$ and the number of local iterations $W$ effect the performance of QUDIO. Besides, we explore the trainability of QUDIO.  Note that for all settings, we fix $K=100$. 

\textbf{The role of $Q$.} Fig.~\ref{fig:app:vqe_ansatz}(b) and (c) show the potential energy surface estimated by QUDIO with the different number of local nodes $Q$. Concisely, for both $Q=2$ and $Q=4$, increasing the number of local steps $W$ incurs an enhanced estimation error. Moreover, when QUDIO synchronizes trainable parameters at every local iteration (e.g., $W=1$), the approximation error approaches to be zero. These outcomes accord with the simulation results obtained in the main text and the claim of Theorem \ref{thm:conv-qnn}. 

\begin{figure*}[htp]
\centering
\includegraphics[width=0.8\textwidth]{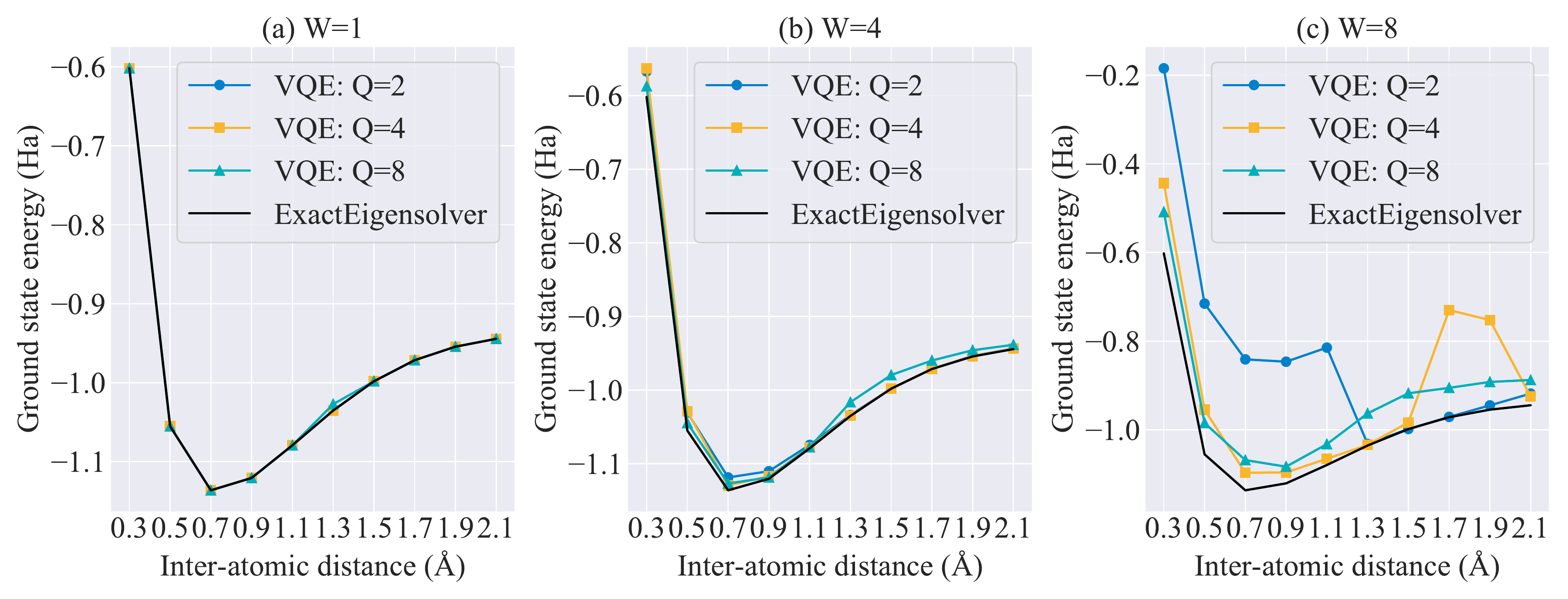}
\caption{\small{\textbf{Performance of QUDIO with different number of local updates $W$.} }}
\label{fig:app:vqe_w}
\end{figure*}

\textbf{The role of $W$.} Fig.~\ref{fig:app:vqe_w} compares the potential energy surface estimated by QUDIO with the different number of local iterations $W$ while the factor of local nodes is set as $Q=\{2, 4, 8\}$. Specifically, when $W=1$, QUDIO achieves the zero approximation error regardless of the number of local nodes $Q$. By contrast, when $W=8$,  the performance of QDUIO becomes inferior. In conjunction with the results in Fig.~\ref{fig:app:vqe_ansatz}, the factor $W$ determines the performance of QUDIO in the NISQ setting, as shown in Theorem \ref{thm:conv-qnn}.

\begin{figure*}[htp]
\centering
\includegraphics[width=1.0\textwidth]{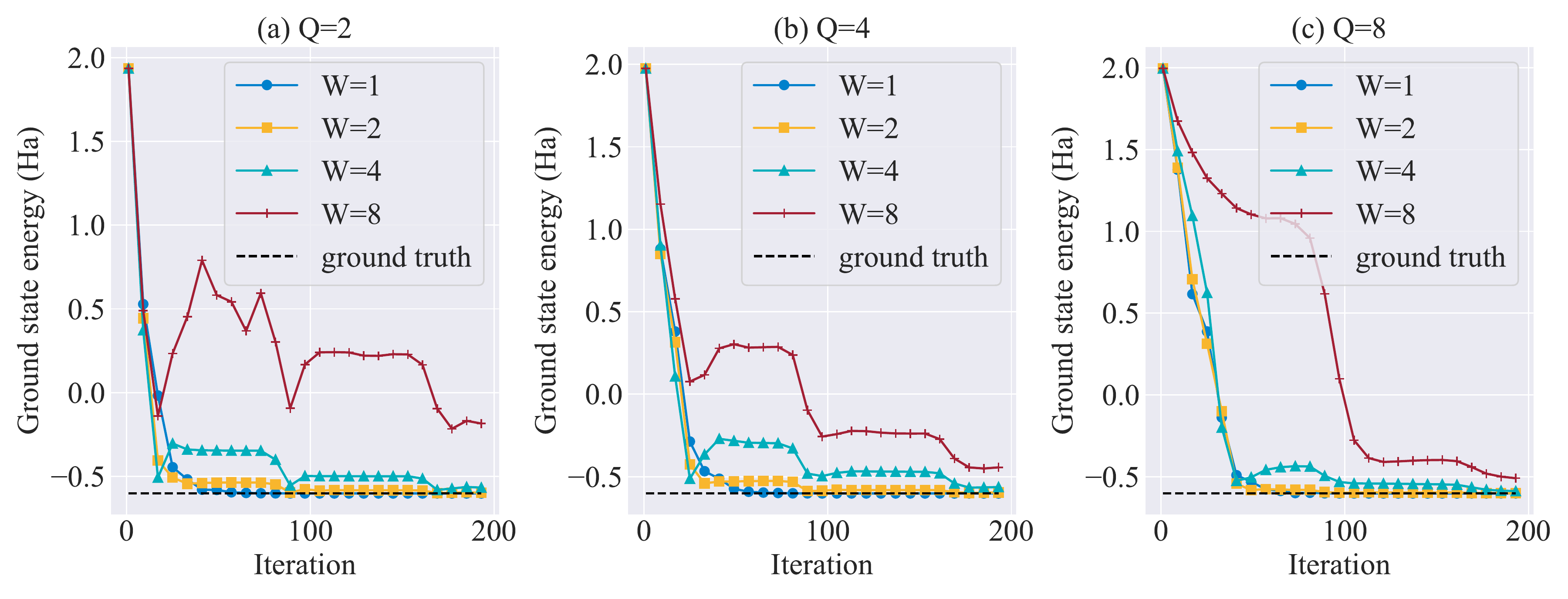}
\caption{\small{\textbf{The estimated energy achieved by  QUDIO during the training process.} The inter-atomic distance is fixed to be  $0.3\mathrm{\AA}$. }}
\label{fig:app:vqe_loss}
\end{figure*}
 
\textbf{The trainability of QUDIO.}   Fig.~\ref{fig:app:vqe_loss} indicates the estimated ground state energy of QUDIO with respect to the number of global iterations. The three subplots hint that the number of local iterations $W$ determines the trainability of QUDIO. Concretely, a smaller number of local updates $W$ assures a faster convergence.   

\end{document}